\documentclass[sigconf]{acmart}

\usepackage[linesnumbered,lined,ruled]{algorithm2e}
\usepackage{subcaption}
\usepackage{color}
\usepackage{todonotes}
\usepackage{paralist}

\usepackage{enumitem}

\graphicspath{ {figures/} }

\makeatletter
\def\thickhline{%
  \noalign{\ifnum0=`}\fi\hrule \@height \thickarrayrulewidth \futurelet
   \reserved@a\@xthickhline}
\def\@xthickhline{\ifx\reserved@a\thickhline
               \vskip\doublerulesep
               \vskip-\thickarrayrulewidth
             \fi
      \ifnum0=`{\fi}}
\newcommand{\removelatexerror}{\let\@latex@error\@gobble}
\makeatother

\newlength{\thickarrayrulewidth}
\begin{document}
\author{Zeyu Ding}
\affiliation{
  \institution{Pennsylvania State University}
}
\email{dxd437@psu.edu}

\author{Yuxin Wang}
\affiliation{
  \institution{Pennsylvania State University}
}
\email{ywang@cse.psu.edu}

\author{Guanhong Wang}
\affiliation{
  \institution{Pennsylvania State University}
}
\email{gpw5092@psu.edu}

\author{Danfeng Zhang}
\affiliation{
  \institution{Pennsylvania State University}
}
\email{zhang@cse.psu.edu}

\author{Daniel Kifer}
\affiliation{
  \institution{Pennsylvania State University}
}
\email{dkifer@cse.psu.edu}

\title{Detecting Violations of Differential Privacy} 
\copyrightyear{2018} 
\acmYear{2018} 
\setcopyright{acmcopyright}
\acmConference[CCS '18]{2018 ACM SIGSAC Conference on Computer and Communications Security}{October 15--19, 2018}{Toronto, ON, Canada}
\acmBooktitle{2018 ACM SIGSAC Conference on Computer and Communications Security (CCS '18), October 15--19, 2018, Toronto, ON, Canada}
\acmPrice{15.00}
\acmDOI{10.1145/3243734.3243818}
\acmISBN{978-1-4503-5693-0/18/10}

\begin{abstract}
The widespread acceptance of differential privacy has led to the publication of many sophisticated algorithms for protecting privacy. However, due to the subtle nature of this privacy definition, many such algorithms have bugs that make them violate their claimed privacy. In this paper, we consider the problem of producing counterexamples for such incorrect algorithms. The counterexamples are designed to be short and human-understandable so that the counterexample generator can be used in the development process -- a developer could quickly explore variations of an algorithm and investigate where they break down. 
Our approach is statistical in nature. It runs a candidate algorithm many times and uses statistical tests to try to detect violations of differential privacy.
An evaluation on a variety of incorrect published algorithms validates the usefulness of our approach: it correctly rejects incorrect algorithms and provides counterexamples for them within a few seconds. 
\end{abstract}




\maketitle

\section{Introduction}\label{sec:intro}
Differential privacy has become a de facto standard for extracting information from a dataset (e.g., answering queries, building machine learning models, etc.) while protecting the confidentiality of individuals whose data are collected. Implemented correctly, it guarantees that any individual's record has very little influence on the output of the algorithm.

However, the design of differentially private algorithms is very subtle and error-prone -- it is well-known that a large number of published algorithms are incorrect (i.e. they violate differential privacy). A sign of this problem is the existence of papers that are solely designed to point out errors in other papers \cite{lyu2017understanding,ashwinsparse}. The problem is not limited to novices who may not understand the subtleties of differential privacy; it even affects experts whose goal is to design sophisticated algorithms for accurately releasing statistics about data while preserving privacy.

There are two main approaches to tackling this prevalence of bugs: programming platforms and verification. Programming platforms, such as PINQ \cite{pinq}, Airavat \cite{Roy10Airavat}, and GUPT \cite{gupt} provide a small set of primitive operations that can be used as building blocks of algorithms for differential privacy. They make it easy to create correct differentially private algorithms at the cost of accuracy (the resulting privacy-preserving query answers and models can become less accurate). Verification techniques, on the other hand, allow programmers to implement a wider variety of algorithms and verify proofs of correctness (written by the developers) \cite{Barthe12, EasyCrypt, BartheICALP2013, Barthe14, Barthe16, BartheCCS16} or synthesize most (or all) of the proofs \cite{Aws:synthesis,lightdp,Fuzz,DFuzz}. 

In this paper, we take a different approach: finding bugs that cause algorithms to violate differential privacy, and generating counterexamples that illustrate these violations. We envision that such a counterexample generator would be useful in the development cycle -- variations of an algorithm can be quickly evaluated and buggy versions could be discarded (without wasting the developer's time in a manual search for counterexamples or a doomed search for a correctness proof). Furthermore, counterexamples can help developers understand why their algorithms fail to satisfy differential privacy and thus can help them fix the problems. This feature is absent in all existing programming platforms and verification tools.
To the best of our knowledge, this is the first paper that treats the problem of detecting counterexamples in incorrect implementations of differential privacy. 

Although recent work on relational symbolic execution \cite{pcg17} aims for simpler versions of this task (like detecting incorrect calculations of sensitivity), it is not yet powerful enough to reason about probabilistic computations. Hence, it cannot 
detect counterexamples in sophisticated algorithms like the \emph{sparse vector technique} \cite{dwork2014algorithmic}, which satisfies differential privacy but is notorious for having many incorrect published variations \cite{lyu2017understanding,ashwinsparse}.

Our counterexample generator is designed to function in black-box mode as much as possible. That is, it executes code with a variety of inputs and analyzes the (distribution of) outputs of the code. This allows developers to use their preferred languages and libraries as much as possible; in contrast, most language-based tools will restrict developers to specific programming languages and a very small set of libraries. In some instances, the code may include some tuning parameters. In those cases, we can use an optional symbolic execution model (our current implementation analyzes python code) to find values of those parameters that make it easier to detect counterexamples. Thus, we refer to our method as a \emph{semi-black-box} approach.

Our contributions are as follows:
\begin{itemize}[leftmargin=5mm]
\item We present the first counterexample generator for differential privacy. It treats programs as semi-black-boxes and uses statistical tests to detect violations of differential privacy.
\item We evaluate our counterexample generator on a variety of sophisticated differentially private algorithms and their common incorrect variations. These include the sparse vector method and noisy max \cite{dwork2014algorithmic}, which are cited as the most challenging algorithms that have been formally verified so far \cite{Aws:synthesis,Barthe16}. In particular, the sparse vector technique is notorious for having many incorrect published variations \cite{lyu2017understanding,ashwinsparse}.  We also evaluate the counterexample generator on some simpler algorithms such as the histogram algorithm \cite{Dwork:2006:DP:2097282.2097284}, which are also easy for novices to get wrong (by accidentally using too little noise). In all cases, our counterexample generator produces counterexamples for incorrect versions of the algorithms, thus showing its usefulness to both experts and novices.
\item The false positive error (i.e. generating "counterexamples" for correct code) of our algorithm is controllable because it is based on statistical testing. The false positive rate can be made arbitrarily small just by giving the algorithm more time to run.
\end{itemize}

Limitations: it is impossible to create counterexample/bug detector that works for all programs. For this reason, our counterexample generator is not intended to be used in an adversarial setting (where a rogue developer wants to add an algorithm that appears to satisfy differential privacy but has a back door). In particular, if a program satisfies differential privacy except with an extremely small probability (a setting known as \emph{approximate differential privacy} \cite{dworkKMM06:ourdata}) then our counterexample generator may not detect it. Solving this issue is an area for future work.

The rest of the paper is organized as follows. Related work is discussed in Section \ref{sec:related}. Background on differential privacy and statistical testing is discussed in Section \ref{sec:background}. The counterexample generator is presented in Section \ref{sec:detect}. Experiments are presented in Section \ref{sec:experiments}. Conclusions and future work are discussed in Section \ref{sec:conc}.

\section{Related Work}\label{sec:related}
\paragraph{Differential privacy} The term differential privacy covers a family of privacy definitions that include pure $\epsilon$-differential privacy (the topic of this paper) \cite{dwork06Calibrating} and its relaxations: approximate $(\epsilon,\delta)$-differential privacy \cite{dworkKMM06:ourdata}, concentrated differential privacy \cite{DR2016:cdp,BS2016:zcdp}, and Renyi differential privacy \cite{M2017:Renyi}. The pure and approximate versions have received the most attention from algorithm designers (e.g., see the book \cite{dwork2014algorithmic}). However, due to the lack of availability of easy-to-use debugging and verification tools, a considerable fraction of published algorithms are incorrect. In this paper, we focus on algorithms for which there is a public record of an error (e.g., variants of the sparse vector method \cite{lyu2017understanding,ashwinsparse}) or where a seemingly small change to an algorithm breaks an important component of the algorithm (e.g., variants of the noisy max algorithm \cite{dwork2014algorithmic,Barthe16} and the histogram algorithm \cite{Dwork:2006:DP:2097282.2097284}).

\paragraph{Programming platforms and verification tools}

Several dynamic tools~\cite{pinq, Roy10Airavat,Tschantz11,Xu2014,
EbadiPOPL2015} exist for enforcing differential privacy.  Those tools track the
privacy budget consumption at runtime, and terminates a program when the
intended privacy budget is exhausted. On the other hand, static methods exist
for verifying that a program obeys differential privacy during any execution,
based on relational program logic~\cite{Barthe12, EasyCrypt, BartheICALP2013,
Barthe14, Barthe16, BartheCCS16, Aws:synthesis} and relational type
system~\cite{lightdp,Fuzz,DFuzz}. We note that those methods are largely orthogonal to
this paper: their goal is to verify a correct program or to terminate an
incorrect one, while our goal is to detect an incorrect program and generate
counterexamples for it. The counterexamples provide valuable guidance for
fixing incorrect algorithms for algorithm designers. Moreover, we believe our
tool fills in the currently missing piece in the development of
differentially private algorithms: with our tool, immature designs can first be
tested for counterexamples, before being fed into those dynamic and static
tools.

\paragraph{Counterexample generation}

Symbolic execution~\cite{king1976symbolic,Cadar06EXE, Cadar08KLEE} is  widely
used for program testing and bug finding. One attractive feature of symbolic
execution is that when a property is being violated, it generates
counterexamples (i.e., program inputs) that lead to violations.  More relevant
to this paper is work on testing relational properties based on symbolic
execution~\cite{person2008, milushev2012, pcg17}. However, those work only
apply to deterministic programs, but the differential privacy property
inherently involves probabilistic programs, which is beyond the scope of those
work.

\section{Background}\label{sec:background}
In this section, we discuss relevant background on differential privacy and hypothesis testing.

\subsection{Differential Privacy}
We view a database as a finite multiset of records from some domain. It is sometimes convenient to represent a database by a histogram, where each cell is the count of times a specific record is present. 

Differential privacy relies on the notion of \emph{adjacent} databases. The two most common definitions of adjacency are: 
\begin{inparaenum}[(1)]
\item  two databases $D_1$ and $D_2$ are \emph{adjacent} if $D_2$ can be obtained from $D_1$ by \emph{adding or removing} a single record.
\item  two databases $D_1$ and $D_2$ are \emph{adjacent} if $D_2$ can be obtained from $D_1$ by \emph{modifying} one record.
\end{inparaenum}
The notion of adjacency used by an algorithm must be provided to the counterexample generator. We write $D_1\sim D_2$ to mean that $D_1$ is adjacent to $D_2$ (under whichever definition of adjacency is relevant in the context of a given algorithm).

We use the term $\emph{mechanism}$ to refer to an algorithm $M$ that tries to protect the privacy of its input. In our case, a mechanism is an algorithm that is intended to satisfy $\epsilon$-differential privacy:

\begin{definition}[Differential Privacy \cite{dwork06Calibrating}]\label{dp} Let $\epsilon\geq 0$. A mechanism $M$ is said to be $\epsilon$-differentially private if for every pair of adjacent databases $D_1$ and $D_2$, and every $E\subseteq \textrm{Range}(M)$, we have 
\[
P(M(D_1)\in E) \leq e^\epsilon \cdot P(M(D_2) \in E).
\]
\end{definition}
The value of $\epsilon$, called the privacy budget, controls the level of the privacy: the smaller $\epsilon$ is, the more privacy is guaranteed.

One of the most common building blocks of differentially private algorithms is the \emph{Laplace mechanism} \cite{dwork06Calibrating} , which is used to answer numerical queries. Let $\mathcal{D}$ be the set of possible databases. A numerical query is a function $q: \mathcal{D}\rightarrow\mathbb{R}^k$ (i.e. it outputs a $k$-dimensional vector of numbers). The Laplace mechanism is based on a concept called \emph{global sensitivity}, which measures the worst-case effect one record can have on a numerical query:
\begin{definition}[Global Sensitivity \cite{dwork06Calibrating}]
The $\ell_1$-global sensitivity of a numerical query $q$ is 
\[\Delta_q = \max_{D_1 \sim D_2} \lVert q(D_1) - q(D_2) \rVert_1.\]
\end{definition}

The Laplace mechanism works by adding Laplace noise (having density
$f(x\mid  b) = \frac{1}{2b}\exp\left(-|x|/b\right)$ and variance  $2b^2$)  to query answers. The chosen variance depends on $\epsilon$ and the global sensitivity.
We use the notation  $\texttt{Lap}(b)$ to refer to the Laplace noise.
\begin{definition}[The Laplace mechanism \cite{dwork06Calibrating}]
For any numerical query $q: \mathcal{D}\rightarrow\mathbb{R}^n$, the Laplace mechanism outputs
\[M(D, q, \epsilon) = q(D) + (\eta_1, \ldots, \eta_n)\]
where $\eta_i$ are independent random variables sampled from $\texttt{Lap}(\Delta_q/\epsilon)$.
\end{definition}

\begin{theorem}[\cite{dwork2014algorithmic}]
The Laplace mechanism is $\epsilon$-differentially private.
\end{theorem}

\subsection{Hypothesis Testing}
\label{sec:hypotest}
A statistical \emph{hypothesis} is a claim about the parameters of the distribution that generated the data.  The \emph{null hypothesis}, denoted by $H_0$  is a  statistical hypothesis that we are trying to disprove. For example, if we have two samples, $X$ and $Y$ where $X$ was generated by a Binomial$(n,p_1)$ distribution and $Y$ was generated by a Binomial$(n, p_2)$ distribution, one null hypothesis could be $p_1=p_2$ (that is, we would like to know if the data supports the conclusion that $X$ and $Y$ came from different distributions). The \emph{alternative hypothesis}, denoted by $H_1$, is the complement of the null hypothesis (e.g., $p_1\neq p_2$).

A hypothesis test is a procedure that takes in a data sample $Z$ and either rejects the null hypothesis or fails to reject the null hypothesis. 
A hypothesis test can have two types of errors: type I and type II.
A  type I error occurs if the test incorrectly rejects $H_0$ when it is in fact true. A type II error occurs if the test fails to reject $H_0$ when the alternative hypothesis is true. Type I and type II errors are analogous to  false positives and false negatives, respectively.

In most problems, controlling type I error is the most important. In such cases, one specifies a \emph{significance level} $\alpha$ and requires that the probability of a type I error be at most $\alpha$. Commonly used values for $\alpha$ are $0.05$ and $0.01$. In order to allow users to control the type I error, the hypothesis test also returns a number $p$ -- known as the p-value -- which is a probabilistic estimate of how unlikely it is that the null hypothesis is true. The user rejects the null hypothesis if $p\leq \alpha$.
In order for this to work (i.e. in order for the Type I error to be below $\alpha$), the $p$-value must satisfy certain technical conditions:
\begin{inparaenum}[(1)]
\item a $p$-value is a function of a data sample $Z$, 
\item $0\leq p(Z)\leq 1$,
\item  if the null hypothesis is true, then $P(p(Z) \leq \alpha~|~H_0)\leq \alpha$.
\end{inparaenum}

A relevant example of a hypothesis test is Fisher's exact test \cite{fisher:1935} for two binomial populations.
Let $c_1$ be a sample from a Binomial$(n_1, p_1)$  distribution and let $c_2$ be a sample from a Binomial$(n_2, p_2)$ distribution. Here $p_1$ and $p_2$ are unknown. Using these values of $c_1$ and $c_2$, the goal is to test the null hypothesis $H_0: p_1\leq p_2$ against the alternative $H_1: p_1>p_2$. Let $s=c_1+c_2$. The key insight behind Fisher's test is that if $C_1\sim $ Binomial$(n_1,p_1)$ \footnote{This is read as "$C_1$ is a random variable having the Binomial$(n_1,p_1)$ distribution".} and $C_2\sim$ Binomial$(n_2, p_2)$ and if $p_1=p_2$, then the value $P(C_1 \geq c_1~|~C_1+C_2=s)$ does not depend on the unknown parameters $p_1$ or $p_2$  and can be computed from the cumulative distribution function of the hypergeometric distribution; specifically, it is equal to $1-\text{Hypergeometric.cdf}(c_1 - 1\mid n_1+n_2, n_1, s)$. When $p_1 < p_2$, then $P(C_1 \geq c_1~|~C_1+C_2=s)$ cannot be computed without knowing $p_1$ and $p_2$. However, it is less than  $1-\text{Hypergeometric.cdf}(c_1 - 1\mid n_1+n_2, n_1, s)$. Thus it  can be shown that $1-\text{Hypergeometric.cdf}(c_1 - 1\mid n_1+n_2, n_1, s)$ is a valid $p$-value and so the Fisher's exact test rejects the null hypothesis when this quantity is $\leq \alpha$.

\section{Counterexample Detection}\label{sec:detect}
For a mechanism $M$ that does not satisfy $\epsilon$-differential privacy, the goal is to prove this failure. By Definition~\ref{dp}, this involves finding a pair of adjacent databases $D_1, D_2$ and an output event $E$ such that $P(M(D_1) \in E) > e^\epsilon P(M(D_2)\in E)$. Thus a counterexample involves finding these two adjacent inputs $D_1$ and $D_2$, the bad output set $E$, and to show that for these choices, $P(M(D_1) \in E) > e^\epsilon P(M(D_2)\in E)$.

Ideally, one would compute the probabilities $P(M(D_1) \in E)$ and $P(M(D_2)\in E)$. Unfortunately, for sophisticated mechanisms, it is not always possible to compute these quantities exactly. However, we can sample from these distributions many times  by repeatedly running $M(D_1)$ and $M(D_2)$ and counting the number of times that the outputs fall into $E$. Then, we need a statistical test to reject the null hypothesis $P(M(D_1) \in E) \leq e^\epsilon P(M(D_2)\in E)$ (or fail to reject it if the algorithm is $\epsilon$-differentially private).

We will be using the following conventions:
\begin{itemize}[leftmargin=5mm]
\item  The input to most mechanisms is actually a list of queries $Q = (q_1, \ldots, q_l)$ rather than a database directly. For example, algorithms to release differentially private histograms operate on a histogram of the data; the sparse vector mechanism operates on a sequence of queries that each have global sensitivity equal to 1.  Thus, we require the user to specify how the input query answers can differ on two adjacent databases. For example, in a histogram, exactly one cell count changes by at most 1. In the sparse vector technique \cite{dwork2014algorithmic}, every query  answer changes by at most 1.
To simplify the discussion, we abuse notation and use $D_1, D_2$ to also denote the answers of $Q$ on the input adjacent databases. For example, when discussing the sparse vector technique,  we write $D_1 = [0,0]$ and $D_2=[1,1]$. This means there are adjacent databases and a list of queries $Q=[q_1, q_2]$ such that they evaluate to $[0,0]$ on the first database and $[1,1]$ on the second database. 

\item We use $\epsilon_0$ to indicate the privacy level that a mechanism claims to achieve.

\item We use $\Omega$ for the set of all possible outputs (i.e., range) of the mechanism $M$. We use $\omega$ for a single output of $M$. 

\item We call a subset $E\subseteq \Omega$ an event. We use $p_1$ (respectively, $p_2$) to denote  $P(M(D_i)\in E)$, the probability that the output of $M$ falls into $E$ when executing on  database $D_1$ (respectively, $D_2$).

\item Some mechanisms take additional inputs, e.g., the sparse vector mechanism. We collectively refer to them as \emph{args}.

\end{itemize}

Our discussion is organized as follows. We provide an overview of the counterexample generator in Section \ref{sec:overview}. Then we incrementally explain our approach. In Section \ref{subsec:hypotest} we present the hypothesis test. That is, suppose we already have query sequences $D_1$ and $D_2$ that are generated from adjacent databases and an output set $E$, how do we test if $P(M(D_1)\in E) \leq e^\epsilon P(M(D_2)\in E)$ or $P(M(D_1)\in E) > e^\epsilon P(M(D_2)\in E)$? Next, in Section
\ref{sec:event_selection}, we consider the question of output selection. That is, suppose we already have query answers $D_1$ and $D_2$ that are generated from adjacent databases, how do we decide which $E$ should be used in the hypothesis test? Finally, in Section \ref{subsec:input}, we consider the problem of generating the adjacent query sequences $D_1$ and $D_2$ as well as additional inputs \emph{args}.

The details of specific mechanisms we test for violations of differential privacy will be given in the experiments in Section \ref{sec:experiments}.

\subsection{Overview}\label{sec:overview}
At a high level, the counterexample generator can be summarized in the pseudocode in Algorithm \ref{alg:detection}. First, it generates an InputList, a set of candidate tuples of the form $(D_1, D_2, args)$. That is, instead of returning a single pair of adjacent inputs $D_1, D_2$ and any auxiliary arguments the mechanism may need, we return multiple candidates which will be filtered later. Each adjacent pair $(D_1, D_2)$ is designed to be short so that a developer can understand the problematic inputs and trace them through the code of the mechanism $M$. For this reason, the code of $M$ will also run fast, so that it will be possible to later evaluate $M(D_1,args)$ and $M(D_2, args)$ multiple times very quickly.

\begin{algorithm}[!ht]
\SetKwProg{Fn}{function}{\string:}{}
\SetKwFunction{Wrapper}{CounterExampleDetection}
\SetKwInOut{Input}{input}
\DontPrintSemicolon
\Fn{\Wrapper{$M$, $\epsilon$}}{
\Input{
$M$: mechanism\\ 
$\epsilon$: desired privacy (is $M$ $\epsilon$-differentially private?)\\
}
\tcp{get set of possible inputs: $(D_1, D_2, args)$}
InputList $\gets$ \texttt{InputGenerator}($M$, $\epsilon$)\;
$E, D_1, D_2, args \gets$ \texttt{EventSelector}($ M, \epsilon,$ InputList)\;
$p_{\top}, p_{\bot} \gets$ \texttt{HypothesisTest}($ M, \epsilon,  D_1, D_2, args, E$)\;\label{line:hypotest}
\Return{$p_{\top}, p_{\bot}$}
}
\caption{Overview of Counterexample Generator \label{alg:detection}}
\end{algorithm}

The next step is the EventSelector. It takes each tuple $(D_1, D_2, args)$ from InputList and runs $M(D_1, args)$ and $M(D_2, args)$ multiple times. Based on the type of the outputs, it generates a set of candidates for $E$. For example, if the output is a real number, then the set of candidates is the set of intervals $(a, b)$. For each candidate $E$ and each tuple $(D_1, D_2, args)$, it counts how many times $M(D_1,args)$ produced an output $\omega\in E$ and how many times $M(D_2,args)$ produced an output in $E$. Based on these results, it picks one specific $E$ and one tuple $(D_1, D_2, args)$ which it believes is most likely to show a violation of $\epsilon_0$-differential privacy.

Finally, the HypothesisTest takes the selected $E$, $D_1$, $D_2$, and args and checks if it can detect statistical evidence that $P(M(D_1, args)\in E) > e^\epsilon P(M(D_2, args)\in E)$ -- which corresponds to the $p$-value $p_{\top}$ -- or $e^\epsilon P(M(D_1, args)\in E) <  P(M(D_2, args)\in E)$ -- which corresponds to the $p$-value $p_{\bot}$.

It is important to note that the EventSelector also uses HypothesisTest internally as a sub-routine to filter out candidates. That is, for every candidate $E$ and every candidate ($D_1$, $D_2$, args), it runs HypothesisTest and treats the returned value as a score. The combination of $E$ and $(D_1, D_2, args)$ with the best score is returned by the EvenSelector. Note that EventSelector is using the HypothesisTest in an exploratory way -- it evaluates many hypotheses and returns the best one it finds. This is why the $E$ and ($D_1$, $D_2$, args) that are finally chosen need to be evaluated again on Line \ref{line:hypotest} using fresh samples from $M$.

\paragraph{Interpreting the results}

\begin{figure*}[!ht]
\begin{subfigure}[b]{0.33\linewidth}
\captionsetup{width=.9\linewidth}
\includegraphics[width=\linewidth]{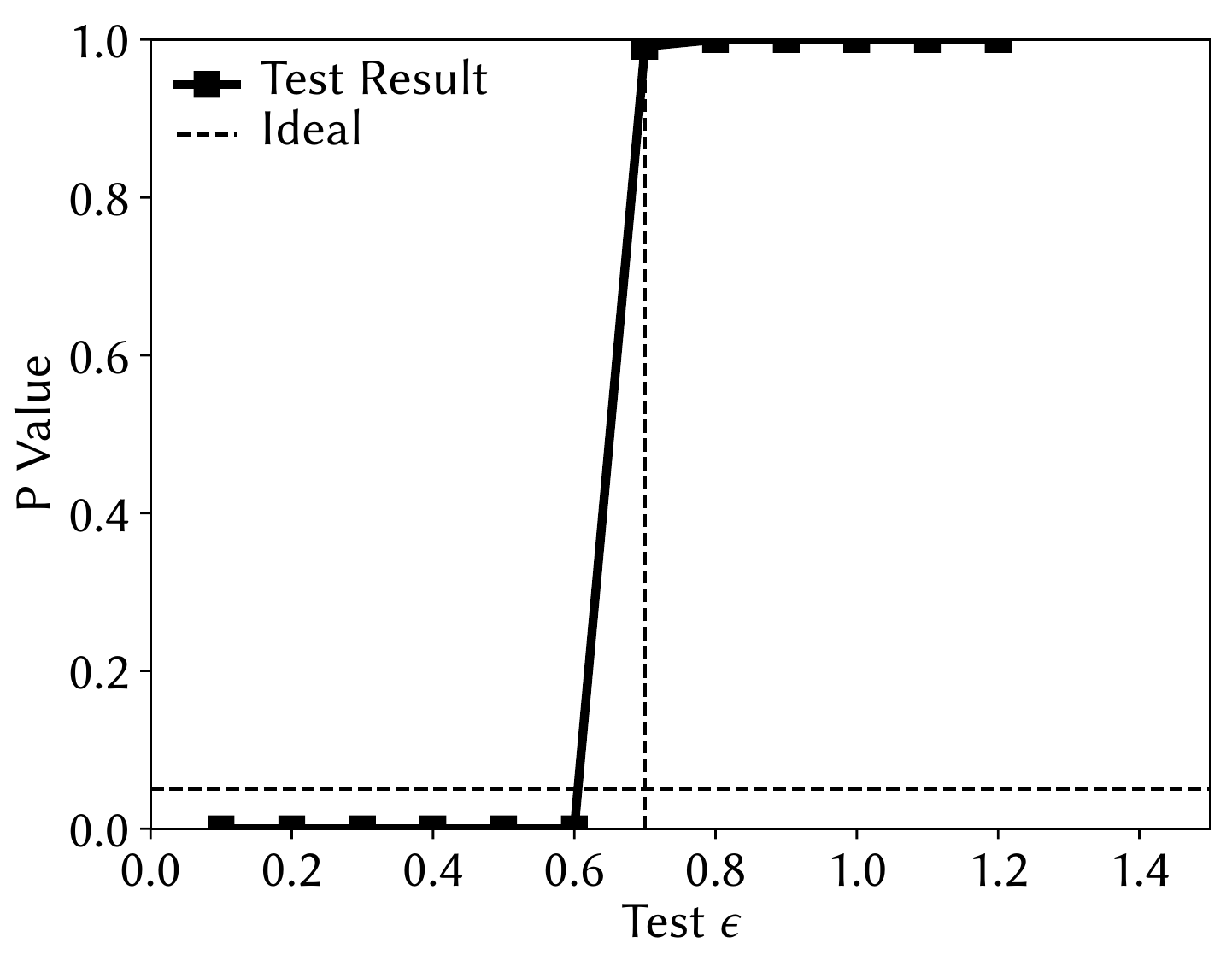}
\caption{Expected results for algorithms correctly claiming $\epsilon_0 = 0.7$ differential privacy.\label{fig:demo1}}
\end{subfigure}
\begin{subfigure}[b]{0.33\linewidth}
\captionsetup{width=.9\linewidth}
\includegraphics[width=\linewidth]{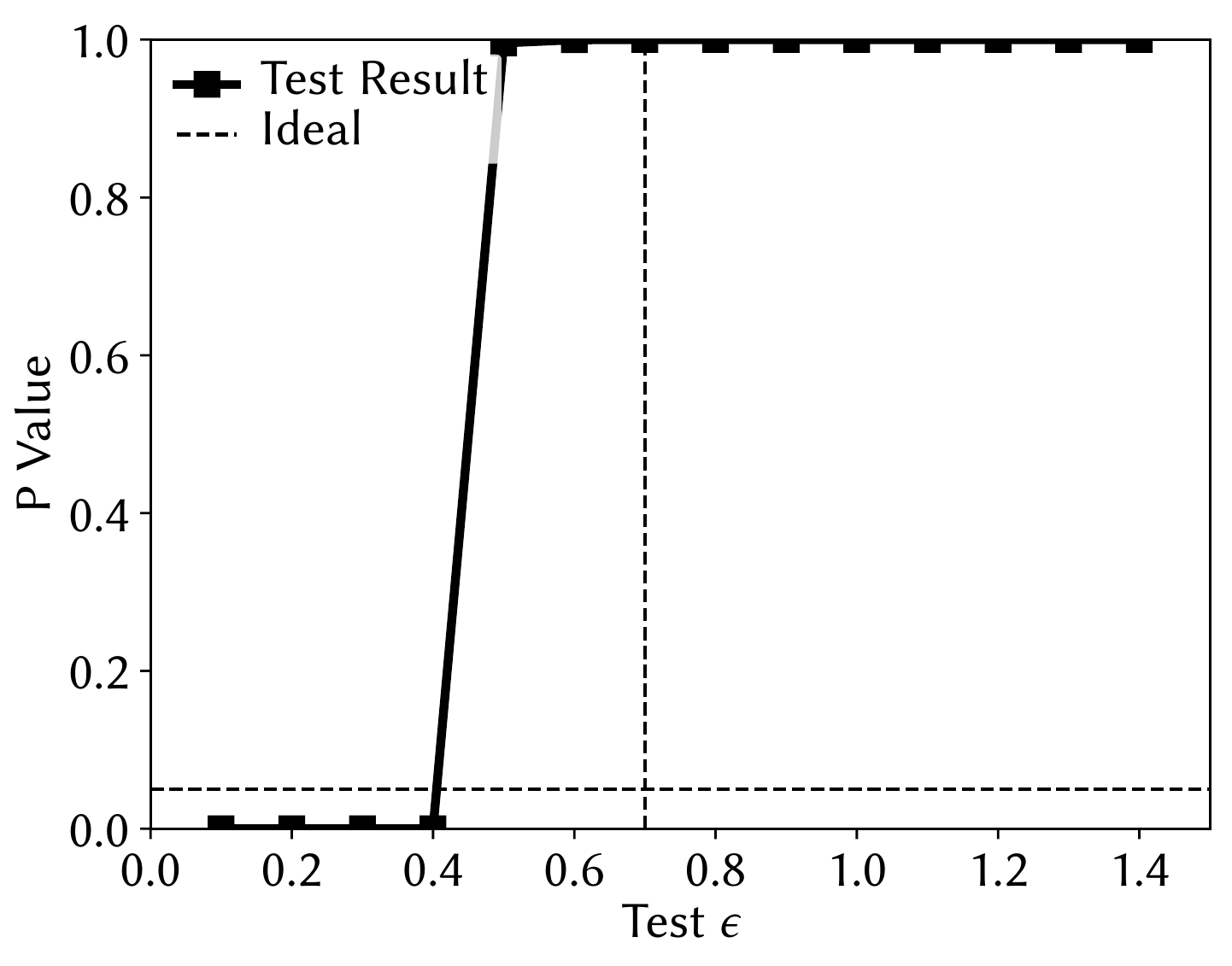}
\caption{Expected results where counterexamples cannot be found or when $M$ satisfies differential privacy for $\epsilon < \epsilon_0=0.7$. \label{fig:demo2}}
\end{subfigure}~
\begin{subfigure}[b]{0.33\linewidth}
\includegraphics[width=\linewidth]{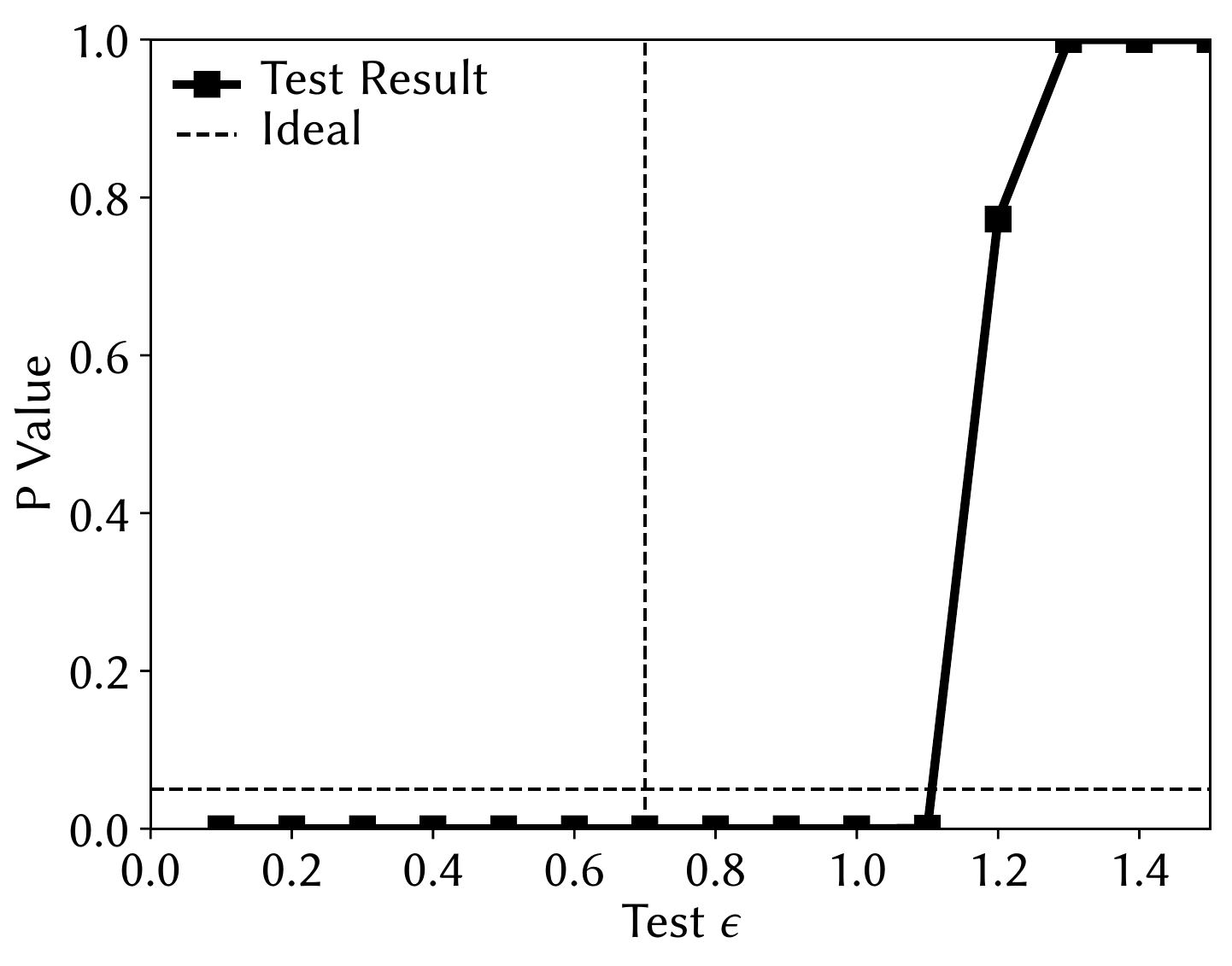}
\caption{Expected results where $M$ does not satisfy $0.7$-differential privacy (i.e. $M$ has a bug and provides less privacy than advertised). \label{fig:demo3}}
\end{subfigure}
\caption{Interpreting experimental results on hypothesis tests. A hypothetical algorithm $M$ claims to achieve $\epsilon_0$-differential privacy. For each $\epsilon$ from $0.2$ to $2.4$ we would evaluate if $M$ satisfies $\epsilon$-differential privacy. We show a typical graph when $M$ does satisfy $\epsilon_0$-differential privacy (left), a graph where $M$ possibly provides more privacy (center) and a graph where $M$ provides less privacy than advertised. \label{fig:demo}}
\end{figure*}

One of the best ways of understanding the behavior of the counterexample generator is to look at the p-values it outputs. That is, we take an mechanism $M$ that claims to satisfy $\epsilon_0$-differential privacy and, for each $\epsilon$ close to $\epsilon_0$, we test whether it satisfies $\epsilon$-differential privacy (that is, even though $M$ claims to satisfy $\epsilon_0$-differential privacy, we may want to test if it satisfies $\epsilon$-differential privacy for some other value of $\epsilon$ that is close to $\epsilon_0$). The hypothesis tester returns two $p$-values: 

\begin{itemize}[leftmargin=5mm]
\item $p_{\top}$. Small values indicate that probably $P(M(D_1, args)\in E) > e^\epsilon P(M(D_2, args)\in E)$.
\item $p_{\bot}$. Small values indicate that probably $P(M(D_2, args)\in E) > e^\epsilon P(M(D_1, args)\in E)$.
\end{itemize}

For each $\epsilon$, we plot the minimum of $p_{\top}$ and $p_{\bot}$. Figure \ref{fig:demo} shows typical results that would appear when the counterexample detector is run with real mechanisms $M$ as input. 

In Figure \ref{fig:demo1}, $M$ correctly satisfies the claimed $\epsilon_0=0.7$ differential privacy. In that plot, we see that the $p$-values corresponding to $\epsilon=0.2, 0.4, 0.6$ are very low, meaning that the counterexample generator can prove that the algorithm does not satisfy differential privacy for those smaller values of $\epsilon$. Near $0.7$ it becomes difficult to find counterexamples; that is, if an algorithm satisfies $0.7$-differential privacy, it is very hard to statistically prove that it does not satisfy $0.65$ differential privacy. This is a typical feature of hypothesis tests as it becomes difficult to reject the null hypothesis when it is only slightly incorrect (e.g., when the true privacy parameter is only slightly different from the $\epsilon$ we are testing). Now, any algorithm that satisfies $0.7$-differential privacy also satisfies $\epsilon$-differential privacy for all $\epsilon \geq 0.7$. This behavior is seen in Figure \ref{fig:demo1} as the $p$-values are large for all larger values of $\epsilon$.

Figure \ref{fig:demo2} shows a graph that can arise from two distinct scenarios. One of the situations is when the mechanism $M$ claims to provide $0.7$-differential privacy but actually provides more privacy (i.e. $\epsilon$-differential privacy for $\epsilon < 0.7$). In this figure, the counterexample generator could prove, for example, that $M$ does not satisfy $0.4$-differential privacy, but leaves open the possibility that it satisfies $0.5$-differential privacy. The other situation is when our tool has failed to find good counterexamples. Thus when a mechanism is correct, good precision by the counterexample generator means that the line starts rising close to (but before the dotted line), and worse precision means that the line starts rising much earlier.

Figure \ref{fig:demo3} shows a typical situation in which an algorithm claims to satisfy $0.7$-differential privacy but actually provides less privacy than advertised. In this case, the counterexample generator can generate good counterexamples at $\epsilon=0.7$ (the dotted line) and even at much higher values of $\epsilon$. When an mechanism is incorrect, such a graph indicates good precision by the counterexample generator.

\paragraph{Limitations}
In some cases, finding counterexamples requires a large input datasets. In those cases, searching for the right inputs and running algorithms on them many times will impact the ability of our counterexample generator to find counterexamples. This is a limitation of all techniques based on statistical tests.

Another important case where our counterexample generator is not expected to perform well is when violations of differential privacy happen very rarely. For example, consider a mechanism $M$ that checks if its input is $D_1=[1]$. If so, with probability $e^{-9}$ it outputs $1$ and otherwise it outputs $0$ (if the input is not $1$, $M$ always outputs $0$). $M$ does not satisfy $\epsilon$-differential privacy for any value of $\epsilon$. However, showing it statistically is very difficult. Supposing $D_1=[1]$ and $D_2=[0]$ are adjacent databases, it requires running $M(D_1)$ and $M(D_2)$ billions of times to observe that an output of $1$ is possible under $D_1$ but is at least $e^\epsilon$ times less likely under $D_2$.

Addressing both of these problems will likely involve incorporation of program analysis, such as symbolic execution, into our statistical framework and is a direction for future work.

\subsection{Hypothesis Testing}\label{subsec:hypotest} 
\begin{algorithm}[!ht]
\SetKwProg{Fn}{function}{\string:}{}
\SetKwFunction{Test}{pvalue}
\SetKwFunction{Wrapper}{HypothesisTest}
\SetKwInOut{Input}{input}
\DontPrintSemicolon
\Fn{\Test{$c_1$, $c_2$, $n$, $\epsilon$}}{
$\tilde{c}_1\gets \text{Binomial}(c_1, 1/e^\epsilon)$\;
s $\gets \tilde{c_1} + c_2$\;
\Return $1 - \text{Hypergeom.cdf}(\tilde{c}_1 - 1 \mid 2n, n, s)$
}
\Fn{\Wrapper{$n$, $M$, $args$, $\epsilon$, $D_1$, $D_2$ , $E$}}{
\Input{
$M$: mechanism\\
$args$: additional arguments for $M$ \\
$\epsilon$: privacy budget to test\\
$D_1$, $D_2$: adjacent databases\\
$E$: Event
}
$\mathcal{O}_1 \gets$ results of running $M(D_1,  args)$ for $n$ times\;
$\mathcal{O}_2 \gets$ results of running $M(D_2, args)$ for $n$ times\;
$c_1 \gets$ $\lvert \{i  \mid \mathcal{O}_1[i] \in E \}\rvert$\;
$c_2 \gets$ $\lvert \{i   \mid \mathcal{O}_2[i] \in E \}\rvert$\;

$p_{\top} \gets$\Test($c_1, c_2, n, \epsilon$)\;
$p_{\bot} \gets$\Test($c_2, c_1, n, \epsilon$)\;

\Return $p_{\top}, p_{\bot}$\;
} 
\caption{Hypothesis Test. Parameter $n$: \# of iterations\label{alg:hypotest}}
\end{algorithm}

Suppose we have a mechanism $M$, inputs $D_1, D_2$ and an output set $E$ (we discuss the generation of $D_1,D_2$ in Section \ref{subsec:input} and $E$ in Section \ref{sec:event_selection}). We would like to check if $P(M(D_1)\in E) > e^\epsilon P(M(D_2)\in E)$ or if  $P(M(D_2)\in E) > e^\epsilon P(M(D_1)\in E)$, as that would demonstrate a violation of $\epsilon$-differential privacy. We treat the  $P(M(D_1)\in E) > e^\epsilon P(M(D_2)\in E)$ case in this section, as the other case is symmetric.

To do this, the high level idea is to:
\begin{itemize}[leftmargin=5mm]
\item Define $p_1 = P(M(D_1)\in E)$ and $p_2=P(M(D_2)\in E)$
\item Formulate the null hypothesis as $H_0: p_1 \leq e^\epsilon \cdot p_2$ and the alternative as $H_1: p_1 > e^\epsilon \cdot p_2$. 

\item Run $M$ with inputs $D_1$ and $D_2$ independently $n$ times each. Record the results as $\mathcal{O}_1$ and $\mathcal{O}_2$.

\item Count the number of times the result falls in $E$ in each case. Let $c_1 = \lvert \{i \mid \mathcal{O}_1[i] \in E \}\rvert$ and $c_2 = \lvert \{i \mid \mathcal{O}_2[i] \in E \}\rvert$. Intuitively, $c_1 \gg e^\epsilon c_2$ provides strong evidence against the null hypothesis.

\item Calculate a $p$-value based on $c_1$, $c_2$ to determine how unlikely the null hypothesis is.
\end{itemize}
The challenge is, of course, in the last step as we don't know what $p_1$ and $p_2$ are. One direction is  to estimate them from $c_1$ and $c_2$. However, it is also challenging to estimate the variance of our estimates $\hat{p}_1$ and $\hat{p}_2$ (the higher the variance, the less the test should trust the estimates). 

Instead, we take a different approach that allows us to conduct the test without knowing what $p_1$ and $p_2$ are.
First, we note that $c_1$ and $c_2$ are equivalent to  samples from a Binomial$(n, p_1)$ distribution and a Binomial$(n, p_2)$ distribution respectively. We first consider the border case where $p_1=e^\epsilon p_2$. Consider sample $\tilde{c}_1$ from a Binomial$(c_1,1/e^{\epsilon})$ distribution. We note that this sample enjoys the following property  (which implies that in the border case, $\tilde{c}_1$ has the same distribution as $c_2$):

\begin{lemma}Let $X\sim $Binomial$(n, p_1)$ and $Z$ be generated from $X$ by sampling from the Binomial$(X, 1/e^\epsilon)$ distribution. The marginal distribution of $Z$ is Binomial$(n, p_1/e^\epsilon)$. %
\end{lemma}
\begin{proof}
The relationship between Binomial and Bernoulli random variables means that $X=\sum_{i=1}^n X_i$, where $X_i$ is a Bernoulli($p_1$)
 random variable. Generating $Z$ from $X$ is the same as doing the following:  set $Z_i=0$ if $X_i=0$. If $X_i=1$, set $Z_i=1$ with probability $1/e^\epsilon$ (and set $Z_i=0$ otherwise). Then set $Z=\sum_{i=1}^n Z_i$. Hence, the marginal distribution of $Z_i$ is a Bernoulli$(p_1/e^\epsilon)$ random variable:
{\small
\begin{align*}
P(Z_i = 1) &= P(Z_i = 1\mid X_i = 1) P(X_i = 1) + P(Z_i = 1\mid X_i=0) P(X_i=0)\\
&= (1/e^\epsilon)\cdot p_1 + 0\cdot (1-p_1)= p_1/e^\epsilon
\end{align*}
}
This means that the marginal distribution of $Z$ is Binomial$(n,p_1/e^\epsilon)$.
\end{proof}

Thus we have the following facts that follow immediately from the lemma:
\begin{itemize}[leftmargin=5mm]
\item If $p_1 > e^\epsilon p_2$ then the distribution of $\tilde{c}_1$ is Binomial$(n, \tilde{p}_1)$ with $\tilde{p}_1= p_1/e^\epsilon$ and so has a larger Binomial parameter than $c_2$ (which is Binomial$(n, p_2))$. We want our test to be able to reject the null hypothesis in this case.
\item If $p_1 = e^\epsilon p_2$ then the distribution of $\tilde{c}_1$ is Binomial$(n, \tilde{p}_1)$ with $\tilde{p}_1=p_2$ and so has the same Binomial parameter as $c_2$. We do not want our test to reject the null hypothesis in this case.
\item If $p_1 < e^\epsilon p_2$ then the distribution of $\tilde{c}_1$ is Binomial$(n, \tilde{p}_1)$ with $\tilde{p}_1=p_1/e^\epsilon$ and so has a smaller Binomial parameter than $c_2$ (which is Binomial$(n, p_2))$. We do not want to reject the null hypothesis in this case.
\end{itemize}

Thus, by randomly generating $\tilde{c}_1$ from $c_1$, we have (randomly) reduced the problem  of 
testing $p_1 > e^\epsilon p_2$ vs. $p_1 \leq e^\epsilon p_2$ (on the basis of $c_1$ and $c_2$) to the problem of testing 
$\tilde{p}_1 > p_2$ vs. $\tilde{p}_1 \leq p_2$ (on the basis of $\tilde{c}_1$ and $c_2$).
Now, checking whether $\tilde{c}_1$ and $c_2$ come from the same distribution can be done with the Fisher's exact test (see Section \ref{sec:background}): the $p$-value is $1 - \text{Hypergeom.cdf}(\tilde{c}_1 - 1\mid 2n, n, \tilde{c}_1+c_2)$.\footnote{Here we use a notation from SciPy \cite{scipy} package where Hypergeom.cdf means the cumulative distribution function of hypergeometric distribution.} This is done in the function \emph{pvalue} in Algorithm \ref{alg:hypotest}.

To summarize, given $c_1$ and $c_2$, we first sample $\tilde{c}_1$ from the Binomial$(c_1, 1/e^\epsilon)$ distribution and then return the p-value of $(1 - \text{Hypergeom.cdf}(\tilde{c}_1 - 1 \mid 2n, n, \tilde{c}_1+c_2))$.
Since this is a random reduction, we reduce its variance by sampling $\tilde{c}_1$ multiple times and averaging the p-values. That is, we run the p-value function (Algorithm \ref{alg:hypotest}) multiple times with the same inputs and average the p-values it returns.

\subsection{Event Selection} \label{sec:event_selection}
Having discussed how to test if $P(M(D_1)\in E) > e^\epsilon P(M(D_2)\in E)$ or if  $P(M(D_2)\in E) > e^\epsilon P(M(D_1)\in E)$ when $D_1,D_2$, and $E$ were pre-specified, we now discuss how to select the event $E$ that is most likely to show violations of $\epsilon$-differential privacy.

\begin{algorithm}[!ht]
\SetKwProg{Fn}{function}{\string:}{}
\SetKwFunction{Wrapper}{EventSelector}
\SetKwInOut{Input}{input}
\DontPrintSemicolon
\Fn{\Wrapper{$n$, $M$, $\epsilon$, InputList}}{
\Input{
$M$: mechanism  \\
$InputList$: possible inputs \\
$\epsilon$: privacy budget to test
}
$pvalues \gets [\ ]$\;
$results \gets [\ ]$\;
\ForEach{$(D_1, D_2, args) \in InputList$} {
$SearchSpace \gets $ search space based on return type\;
$\mathcal{O}_1 \gets$ results of running $M(D_1,  args)$ for $n$ times\;
$\mathcal{O}_2 \gets$ results of running $M(D_2, args)$ for $n$ times\;
\ForEach{$E \in SeachSpace$}{
$c_1 \gets$ $\lvert \{i  \mid \mathcal{O}_1[i] \in E \}\rvert$\;
$c_2 \gets$ $\lvert \{i   \mid \mathcal{O}_2[i] \in E \}\rvert$\;

$p_{\top} \gets$\Test($c_1, c_2, n, \epsilon$)\;\label{line:start}
$p_{\bot} \gets$\Test($c_2, c_1, n, \epsilon$)\;

$pvalues$.append($min$($p_{\top}$, $p_{\bot}$))\;
$results$.append($D_1, D_2, args, E$)\;\label{line:end}
}
}

\Return{$results$[argmin(pvalues)]}\;
}
\caption{Event Selector. Parameter $n$: \# of iterations \label{alg:event_selector}}
\end{algorithm}

One of the challenges is that different mechanisms could have different output types (e.g., a discrete number, a vector of numbers, a vector of categorical values, etc.). To address this problem, we define a search space $S$ of possible events to look at. The search space depends on the type of the output $\omega$ of $M$, which can be determined by running $M(D_1)$ and $M(D_2)$ multiple times.

\begin{enumerate}[leftmargin=5mm]

\item \textbf{The output $\omega$ is a fixed length list of categorical values.} 
We first run $M(D_1)$ once and ask it to not use any noise (i.e. tell it to satisfy $\epsilon$-differential privacy with $\epsilon=\infty$). Denote this output as $\omega_0$. Now, when $M$ runs with its preferred privacy settings to produce an output $\omega$, we define $t(\omega)$ be the Hamming distance between the output $\omega$ and  $\omega_0$. The search space is \[S = \{\{\omega \mid t(\omega) = k\} : k= 0, 1, \ldots, l \}\] where $l$ is the fixed length of output of $M$.
Another set of events relate to the count of a categorical value in the output. If there are $m$ values, then define
\[
S_i = \{\{\omega\mid\omega.count(value_i) = k\}: k= 0, 1, \ldots, l\},
\]
$1\leq i\leq m$. The overall search space is the union of $S$ and all $S_i$.

\item \textbf{The output $\omega$ is a variable length list of categorical values. }
In this case, one extra set of events $E$ we look at correspond to the length of the output. For example, we may check if $P(M(D_1) \text{ has length k}) > P(M(D_2) \text{ has length k})$. Hence, we define 
\[S_0 = \{\{\omega \mid \omega.length = k\} : k= 0, 1, \ldots \}\] 
For the search space $S$, we use this $S_0$ unioned with the search space from the previous case.

\item \textbf{The output $\omega$ is a fixed length list of numeric values. }

In this case, the output is of the form $\omega = (a_1, \ldots, a_m)$. Our search space is the union of the following: 
\begin{align*}
&\{\{\omega\mid\omega[i] \in (a,b)\}: i=1,\dots, m \text{ and } a<b\}\},\\
&\{\{\omega\mid\text{avg}(\omega) \in (a,b)\}:  a<b\}\},\\
&\{\{\omega\mid\text{min}(\omega) \in (a,b)\}:  a<b\}\},\\
&\{\{\omega\mid\text{max}(\omega) \in (a,b)\}:  a<b\}\}.
\end{align*}
That is,  we would end up checking if $P(avg(M(D_1)) \in (a,b)) > P(avg(M(D_2)) \in (a,b))$, etc.. To save time, we often restrict $a$ and $b$ to be multiples of a small number like $\pm0.2$, or $\pm\infty$. In the case that the output $\omega$ is always an integer array, we replace the condition ``$\in (a,b)$'' with ``$=k$'' for each integer $k$.
\item \textbf{$M$ outputs a variable length list of numeric values.} 

The search space is the union of Case 3 and $S_0$ in Case 2.

\item \textbf{$M$ outputs a variable length list of mixed categorical and numeric values. }
In this case, we separate out the categorical values from numeric values and use the cross product of the search spaces for numeric and categorical values. For instance, events would be of the form ``$\omega$  has $k$ categorical components equal to $\ell$ and the average of the numerical components of $\omega$ is in $(a,b)$''

\end{enumerate}

The EventSelector is designed to return one event $E$ for use in the hypothesis test in Algorithm \ref{alg:detection}. The way EventSelector works is it receives an InputList, which is a set of tuples $(D_1, D_2, args)$ where $D_1,D_2$ are adjacent databases and args is a set of values for any other parameters $M$ needs. For each such tuple, it runs $M(D_1)$ and $M(D_2)$ for $n$ times each. Then  for each possible event in the search space, it runs the hypothesis test (as an exploratory tool) to get a p-value. The combination of $(D_1,D_2,args)$ and $E$ that produces the lowest p-value is then returned to Algortihm \ref{alg:detection}. Algorithm \ref{alg:detection} uses those choices to run the real hypothesis test on fresh executions of $M$ on $D_1$ and $D_2$.

The pseudocode for the EventSelector is shown in Algorithm \ref{alg:event_selector}.\footnote{In practice, to avoid choosing bad $E$, we let $c_E$ be the total number of times $M(D_1)$ and/or $M(D_2)$ produced an output in $E$. Then it only executes Line \ref{line:start}-\ref{line:end} in Algorithm \ref{alg:event_selector} if $c_E \geq 0.001 \cdot n\cdot  e^\epsilon$, otherwise the selection of $E$ is too noisy.}

\subsection{Input Generation}\label{subsec:input}

\begin{algorithm}[h]
\SetKwProg{Fn}{function}{\string:}{}
\SetKwFunction{InputGenerator}{InputGenerator}
\SetKwFunction{ArgumentGenerator}{ArgumentGenerator}
\SetKwInOut{Input}{input}
\DontPrintSemicolon

\Fn{\ArgumentGenerator{$M$, $D_1, D_2$}}{
$args_0 \gets$ Arguments used in noise generation with values that minimize the noises\; 
$constraints \gets$ Traverse the source code of $M$ and generate constraints to force $D_1$ and $D_2$ to diverge on branches\;
$args_1 \gets$ \texttt{MaxSMT(constraints)}\;
\Return $args_0 + args_1$
}

\Fn{\InputGenerator{$M$, $len$}}{
\Input{
$M$: mechanism\\
$len$: length of input to generate
}
$candidates \gets$ Empirical pairs of databases of length $len$\;

$InputList \gets [\ ]$\;

\ForEach{$(D_1, D_2) \in candidates$}{
$args \gets \texttt{ArgumentGenerator}(M, D_1, D_2)$\;
$InputList$.append($D_1, D_2, args$)\;
}

\Return{$InputList$}
}
\caption{Input Generator.\label{algo:database_generator}}
\end{algorithm}

In this section we discuss our approaches for generating candidate tuples $(D_1, D_2, args)$ where $D_1,D_2$ are adjacent databases and args is a set of auxiliary parameters that a mechanism $M$ may need.

\subsubsection{Database Generation \label{sec:query_generation}}
To find the adjacent databases that are likely to form the basis of counterexamples that illustrate violations of differential privacy, we adopt a simple and generic approach that works surprisingly well.  Recalling that the inputs to mechanisms are best modeled as a vector of query answers, we use the type of patterns shown in Table \ref{tab:query_category}.

\begin{table}[ht]
\caption{Database categories and samples \label{tab:query_category}}
\begin{center}
\begin{tabular}{c c c c} 
\thickhline
\textbf{Category} & \textbf{Sample D1} & \textbf{Sample D2} \\
\hline
One Above & [1, 1, 1, 1, 1] & [2, 1, 1, 1, 1] \\
One Below & [1, 1, 1, 1, 1] & [0, 1, 1, 1, 1] \\
One Above Rest Below &  [1, 1, 1, 1, 1] & [2, 0, 0, 0, 0] \\
One Below Rest Above & [1, 1, 1, 1, 1] & [0, 2, 2, 2, 2]\\
Half Half & [1, 1, 1, 1, 1] & [0, 0, 0, 2, 2] \\
All Above \& All Below & [1, 1, 1, 1, 1] & [2, 2, 2, 2, 2]\\
X Shape & [1, 1, 0, 0, 0] & [0, 0, 1, 1, 1]\\
\thickhline
\end{tabular}
\end{center}
\end{table}

The ``One Above'' and ``One Below'' categories are suitable for algorithms whose input is a histogram (i.e. in adjacent databases, at most one query can change, and it will change by at most 1). The rest of the categories are suitable when in adjacent databases every query can change by at most one (i.e. the queries have sensitivity\footnote{For queries with larger sensitivity, the extension is obvious. For example $D_1=[1,1,1,1,1]$ and $D_2=[1+\Delta_q, 1+\Delta_q, \dots, 1+\Delta_q]$} $\Delta_q=1$).

The design of the categories is based on the wide variety of changes in query answers that are possible when evaluated on one database and on an adjacent database. For example, it could be the case that a few of the queries increase (by 1, if their sensitivity is 1, or by $\Delta_q$ in the general case) but most of them decrease. A simple representative of this situation is ``One Above Rest Below'' in which one query increases and the rest decrease. The category ``One Below Rest Above'' is the reverse.

Another situation is where roughly half of the queries increase and half decrease (when evaluated on a database compared to when evaluated on an adjacent database). This scenario is captured by the ``Half Half'' category. Another situation is where all of the queries increase. This is captures by the ``All Above \& All Below'' category. Finally, the ``X Shape'' category captures the setting where the query answers are not all the same and some increase and others decrease when evaluated on one database compared to an adjacent database.

These categories were chosen from our desire to allow counterexamples to be easily understood by mechanism designers (and to make it easier for them to manually trace the code to understand the problems). Thus the samples are short and simple. We consider inputs of length 5 (as in Table \ref{tab:query_category}) and also versions of length 10.

\subsubsection{Argument Generation}
Some differentially-private algorithms require extra parameters beyond the
database. For example, the sparse vector technique \cite{dwork2014algorithmic},
shown in Algorithm \ref{algo:sparse_vector_lyu}, takes as inputs a threshold
$T$ and a bound $N$. It tries to output numerical queries that are larger than
$T$. However, for privacy reasons, it will stop after it returns $N$ noisy
queries whose values are greater than $T$. These two arguments are specific to
the algorithm and their proper values depend on the desired privacy level as
well as algorithm precision.

To find values of auxiliary parameters (such as $N$ and $T$ in Sparse Vector),
we build argument generator based on \textit{Symbolic Execution}
\cite{king1976symbolic}, which is typically used for bug finding: it generates
concrete inputs that violate assertions in a program. In general, a symbolic
executor assigns symbolic values, rather than concrete values as normal
execution would do, for inputs. As the execution goes, the executor maintains a
symbolic program state at each assertion and generates constraints that will
violate the assertion. When those constraints are satisfiable, concrete inputs
(i.e., a solution of the constraints) are generated.

Compared with standard symbolic execution, a major difference in our argument
generation is that we are interested in algorithm arguments that will likely
maximize the privacy cost of an algorithm. In other words, there is no obvious
assertion to be checked in our argument generation. To proceed, we use two
heuristics that likely will cause large privacy cost of an algorithm:
\begin{itemize}[leftmargin=5mm]
\item The first heuristic applies to parameters that affect noise generation.
For example in Sparse Vector, the algorithm adds Lap($2 \cdot N \cdot
\Delta_q/\epsilon_0$) noise. For such a variable, we use the value that results
in small amount of noise (i.e., $N=1$). Small amount of noise is favorable
since it reduces the variance in the hypothesis testing
(Section~\ref{subsec:hypotest}).

\item The second heuristic (for variables that do not affect noise)  prefers
arguments that make two program executions using two different databases (as
described in Section~\ref{sec:query_generation}) to take as many diverging
branches as possible. The reason is that diverging branches will likely use
more privacy budget. 
\end{itemize}

Next, we give a more detailed overview of our customized symbolic executor.
The symbolic executor takes a pair of \emph{concrete} databases as inputs (as
described in Section~\ref{sec:query_generation}) and uses \emph{symbolic}
values for other input parameters.  Random samples in the program (e.g., a
sample from Laplace distribution) are set to value 0 in the symbolic execution.
Then, the symbolic executor tracks symbolic program states along program
execution in the standard way~\cite{king1976symbolic}. For example, the
executor will generate a constraint\footnote{For simplicity, we use a simple
representation for constraints; Z3 has an internal format and a user can either
use Z3's APIs or SMT2 \cite{ranise2006smt} format to represent constraints.} $x
= y+1$ after an assignment  (x $\leftarrow$ y+1), assuming that variable y has a symbolic
value $y$ before the assignment. Also, the executor will unroll loops in the
source code, which is standard in most symbolic executors.

Unlike standard symbolic executors, the executor conceptually tracks a pair of
symbolic program states along program execution (one on concrete database
$D_1$, and one on concrete database $D_2$). Moreover, it also generate extra
constraints, according to the two heuristics above, in the hope of maximizing
the privacy cost of an algorithm. In particular, it handles two kinds of
statements in the following way:
\begin{itemize}[leftmargin=5mm]
\item \textbf{Sampling.} The executor generates two constraints for a sampling
statement: a constraint that eliminates randomness in symbolic execution by
assigning sample to value 0, and a constraint that ensures a small amount of
noise. Consider a statement ($\eta\ \leftarrow\ Lap(e)$). The executor
generates two constraints: $\eta = 0$ as well as a constraint that minimizes
expression $e$.

\item \textbf{Branch.} The executor generates a constraint that makes the two
executions diverge on branches. Consider a branch statement ($\textbf{if}\
\text{e}\ \textbf{then} \cdots$). Assume that the executor has symbolic values
$e_1$ and $e_2$ for the value of expression e on databases $D_1$ and $D_2$
respectively; it will generates a constraint $(e_1 \land \neg e_2) \lor (\neg
e_1 \land e_2)$ to make the executions diverge. Note that unlike other
constraints, a diverging constraint might be unsatisfiable (e.g., if the
query answers under $D_1$ and $D_2$ are the same). However, our goal is to
\emph{maximize} the number of satisfiable diverging constraints, which can be
achieved by a MaxSMT solver.
\end{itemize}

The executor then uses an external MaxSMT solver such as Z3 \cite{de2008z3} on
all generated constraints to find arguments that maximizes the number of
diverged branches. 

For example, the correct version of the Sparse Vector algorithm (see the complete algorithm in Algorithm \ref{algo:sparse_vector_lyu}) has the parameter $T$ (a threshold). It has a branch that tests whether the noisy query answer is above the threshold $T$:
\[q+\eta_2 \geq \hat{T}\]
Here, $\eta_2$ is a noise variable, $q$ is one query answer (i.e. one of the components of the input $D_1$ of the algorithm) and $\hat{T}$ is a noisy threshold ($\hat{T}=T+\eta_1$). Suppose we start from a database candidate ([1, 1, 1, 1, 1], [2, 2, 2, 2, 2]). The symbolic executor assigns symbolic values to the parameters $T$ and unrolls the loop in the algorithm, where each iteration handles one noisy query. Along the execution, it updates program states. For example, statement $\hat{T} \gets T + \eta_1$ results in  $\hat{T} = T + \eta_1$.
For the first execution of the branch of interest, the executor tracks the following symbolic program state: 
\[q_1=1 \wedge q_2 = 2 \wedge \eta_1=0 \wedge \eta_2=0 \wedge \hat{T_1}=T+\eta_1 \wedge \hat{T_2}=T+\eta_2\]
as well as the following constraint for diverging branches:
\[(q_1 + \eta_1 \geq \hat{T_1} \wedge q_2 +\eta_2 < \hat{T_2})
\lor (q_1 + \eta_1 < \hat{T_1} \wedge q_2 +\eta_2 \geq \hat{T_2})\] 

Similarly, the executor generates constraints from other iterations. In this example, the MaxSMT solver returns a value in between of $1$ and $2$ so that constraints from all iterations are satisfied. This value of $T$ is used as arg in the candidate tuple $(D_1, D_2, arg)$.

\section{Experiments}\label{sec:experiments}
We implemented our counterexample detection framework with all components,
including hypothesis test, event selector and input generator. The
implementation is publicly available\footnote{
\url{https://github.com/cmla-psu/statdp}.}. The tool takes in an algorithm
implementation and the desired privacy bound $\epsilon_0$, and generates
counterexamples if the algorithm does not satisfy $\epsilon_0$-differential
privacy.

In this section we evaluate our detection framework on some of the popular privacy mechanisms and their variations. We demonstrate the power of our tool: for mechanisms that falsely claim to be differentially private, our tool produces convincing evidence that this is not the case in just a few seconds.

\begin{figure*}
\begin{subfigure}[b]{0.45\textwidth}
\captionsetup{width=.9\linewidth}
\includegraphics[width=\linewidth]{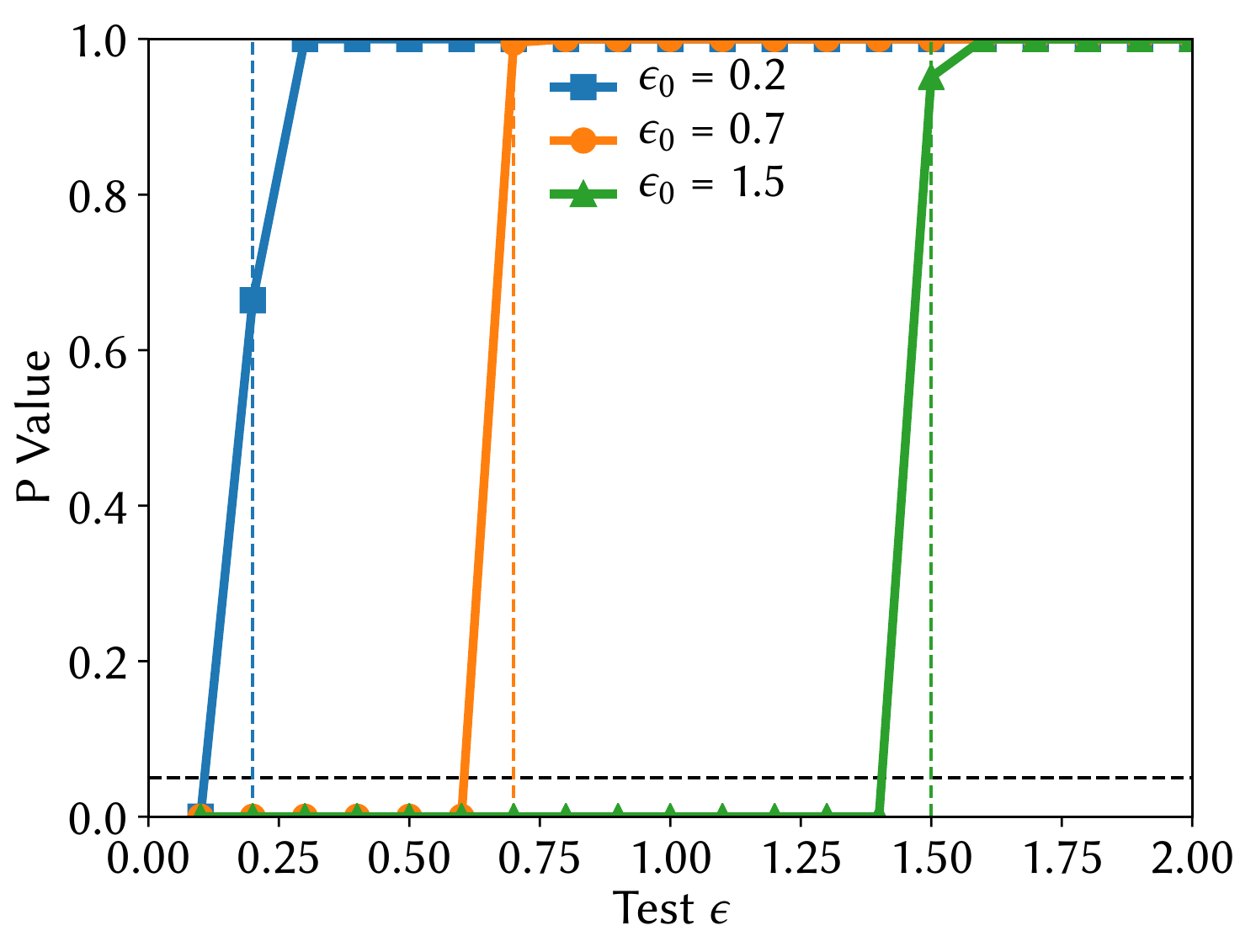}
\caption{Correct Noisy Max with \textit{Laplace} noise\label{fig:noisy_max_v1a}.}
\end{subfigure}
\begin{subfigure}[b]{0.45\textwidth}
\captionsetup{width=.9\linewidth}
\includegraphics[width=\linewidth]{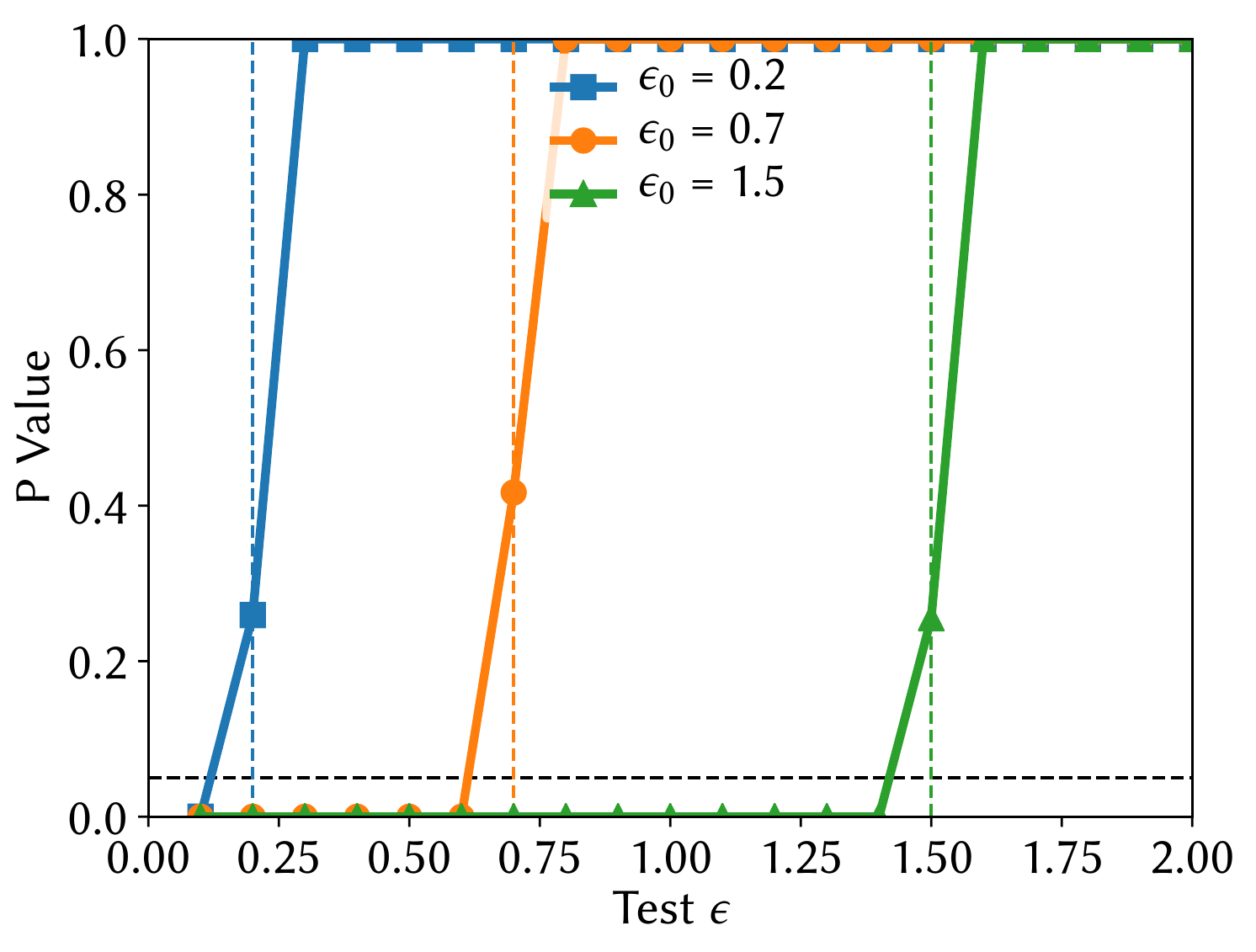}
\caption{Correct Noisy Max with \textit{Exponential} noise\label{fig:noisy_max_v2a}.}
\end{subfigure}
\begin{subfigure}[b]{0.45\textwidth}
\captionsetup{width=.9\linewidth}
\includegraphics[width=\linewidth]{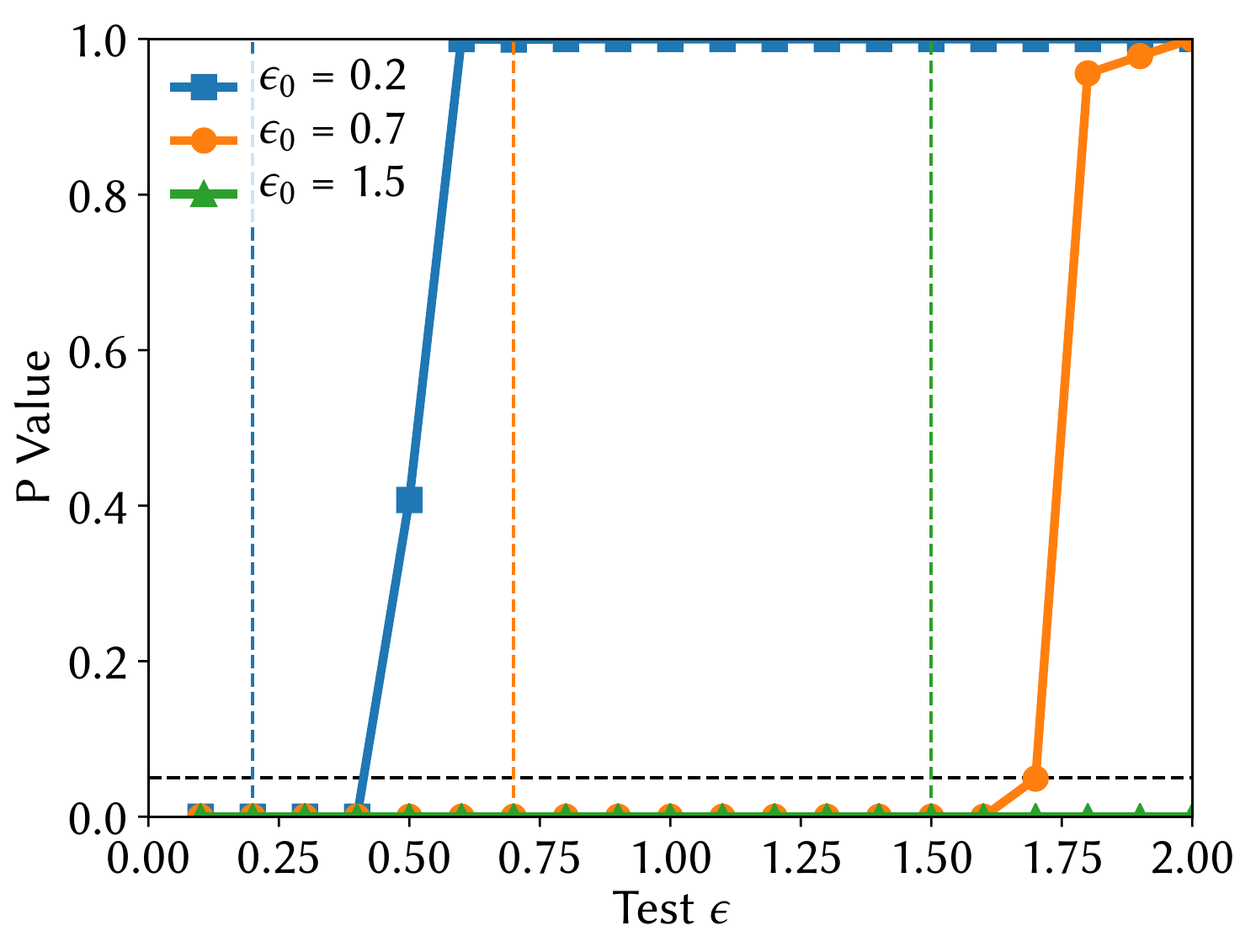}
\caption{Incorrect variant with \textit{Laplace} noise. It returns the maximum value instead of the index. \label{fig:noisy_max_v1b}}
\end{subfigure}
\begin{subfigure}[b]{0.45\textwidth}
\captionsetup{width=.9\linewidth}
\includegraphics[width=\linewidth]{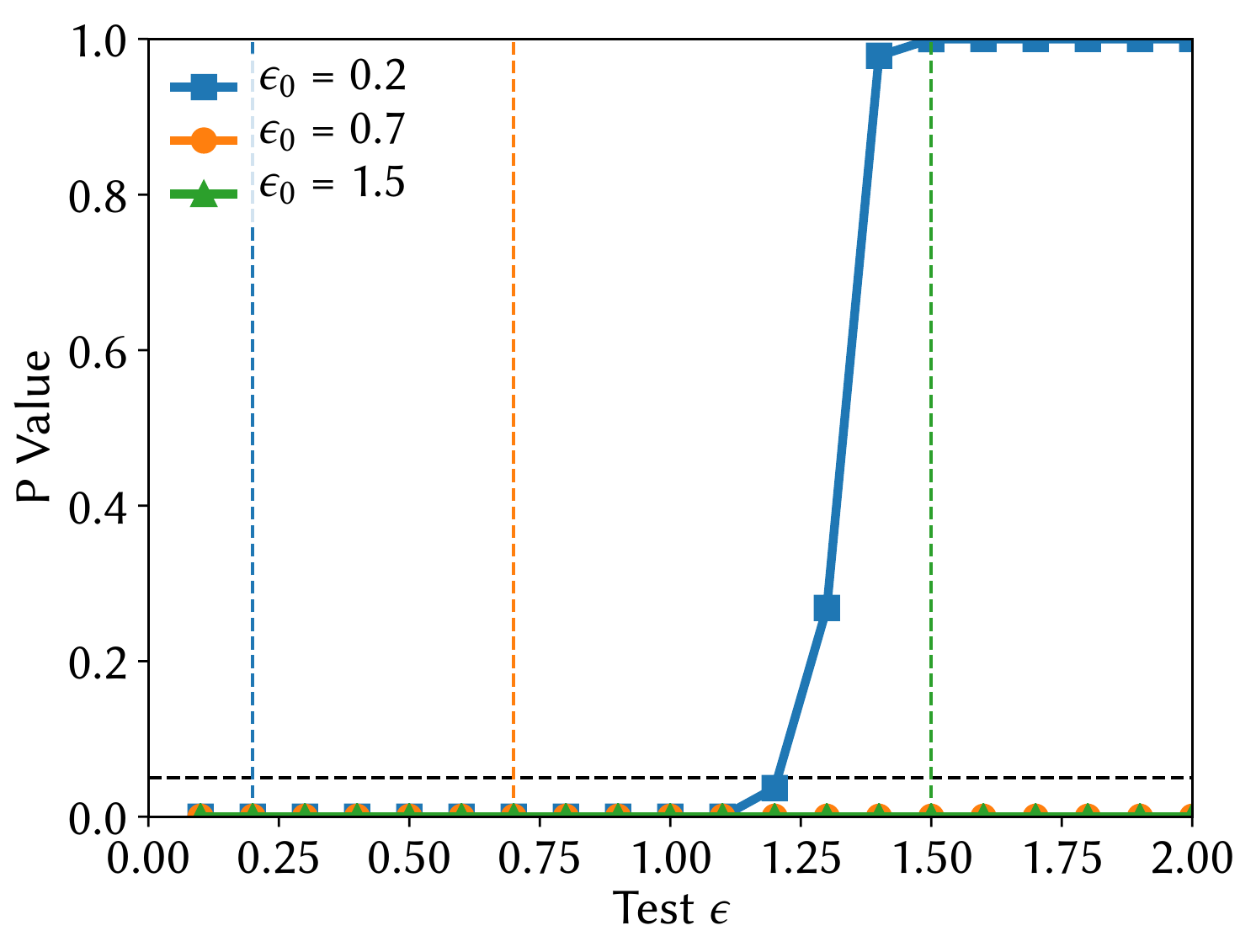}
\caption{Incorrect variant with \textit{Exponential} noise. It returns the maximum value instead of the index. \label{fig:noisy_max_v2b}}
\end{subfigure}
\caption{Results of Noisy Max algorithm and its variants. \label{fig:noisy_max}}
\end{figure*}

\subsection{Noisy Max}

\textit{Report Noisy Max} reports which one among a list of counting queries has the largest value. It adds $Lap(2/\epsilon_0)$ noise to each answer and returns the \textit{index} of the query with the largest noisy answer. The correct versions have been proven to satisfy $\epsilon_0$-differential privacy \cite{dwork2014algorithmic} no matter how long the input list is. A naive proof would show that it satisfies $(\epsilon_0\cdot |Q|/2)$-differential privacy (where $|Q|$ is the length of the input query list), but a clever proof shows that it actually satisfies $\epsilon_0$-differential privacy.

\begin{algorithm}[ht]
\SetKwProg{Fn}{function}{\string:}{}
\SetKwFunction{Test}{NoisyMax}
\SetKwInOut{Input}{input}
\DontPrintSemicolon
\Fn{\Test{$Q$, $\epsilon_0$}}{
\Input{$Q$: queries to the database, $\epsilon_0$: privacy budget.}
NoisyVector $\gets$ $[\ ]$\\
\For{$i=1\dots\text{len(Q)}$}{
	NoisyVector$[i]$ $\gets Q[i]+\textrm{Lap}(2/\epsilon_0)$
}
\Return $argmax$(NoisyVector)
}
\caption{Correct Noisy Max with Laplace noise\label{alg:noisy_max_v1a}}
\end{algorithm}

\subsubsection{Adding $Laplace$ Noise\label{subsec:laplace_noisymax}}
The correct Noisy Max algorithm (Algorithm \ref{alg:noisy_max_v1a}) adds independent $\text{Lap}(2/\epsilon_0)$ noise to each query answer and returns the index of the maximum value. As Figure \ref{fig:noisy_max_v1a} shows, we test this algorithm for different privacy budget $\epsilon_0$ at $0.2,\ 0.7,\ 1.5$. All lines rise when the test $\epsilon$ is slightly less than the claimed privacy level $\epsilon_0$ of the algorithm. This demonstrates the precision of our tool: before $\epsilon_0$, there is almost 0 chance to falsely claim that this algorithm is not private; after $\epsilon_0$, the $p$-value is too large to conclude that the algorithm is incorrect. We note that the test result is very close to the ideal cases, illustrated by the vertical dashed lines.

\begin{algorithm}[ht]
\SetKwProg{Fn}{function}{\string:}{}
\SetKwFunction{Test}{NoisyMax}
\SetKwInOut{Input}{input}
\DontPrintSemicolon
\Fn{\Test{$Q$, $\epsilon_0$}}{
\Input{$Q$: queries to the database, $\epsilon_0$: privacy budget.}
NoisyVector $\gets$ $[\ ]$\\
\For{$i=1\dots\text{len(Q)}$}{
	NoisyVector$[i]$ $\gets Q[i]+\text{Exponential}(2/\epsilon_0)$
}
\Return $argmax$(NoisyVector)
}
\caption{Correct Noisy Max with Exponential noise\label{alg:noisy_max_v2a}}
\end{algorithm}

\begin{algorithm}[t]
\SetKwProg{Fn}{function}{\string:}{}
\SetKwFunction{Test}{NoisyMax}
\SetKwInOut{Input}{input}
\DontPrintSemicolon
\Fn{\Test{$Q$, $\epsilon_0$}}{
\Input{$Q$: queries to the database, $\epsilon_0$: privacy budget.}
NoisyVector $\gets$ $[\ ]$\\
\For{$i=1\dots\text{len(Q)}$}{
	NoisyVector$[i]$ $\gets Q[i]+\text{Laplace}(2/\epsilon_0)$
}
// returns maximum value instead of index \\
\fbox{\Return $max$(NoisyVector)}}
\caption{Incorrect Noisy Max with Laplace noise, returning the maximum value \label{alg:noisy_max_v1b}}
\end{algorithm}

\begin{algorithm}[t]
\SetKwProg{Fn}{function}{\string:}{}
\SetKwFunction{Test}{NoisyMax}
\SetKwInOut{Input}{input}
\DontPrintSemicolon
\Fn{\Test{$Q$, $\epsilon_0$}}{
\Input{$Q$: queries to the database, $\epsilon_0$: privacy budget.}
NoisyVector $\gets$ $[\ ]$\\
\For{$i=1\dots\text{len(Q)}$}{
	NoisyVector$[i]$ $\gets Q[i]+\text{Exponential}(2/\epsilon_0)$
}
// returns maximum value instead of index \\
\fbox{\Return $max$(NoisyVector)}
}
\caption{Incorrect Noisy Max with Exponential noise, returning the maximum value \label{alg:noisy_max_v2b}}
\end{algorithm}

\begin{algorithm}[t]
\SetKwProg{Fn}{function}{\string:}{}
\SetKwFunction{Test}{Histogram}
\SetKwInOut{Input}{input}
\DontPrintSemicolon
\Fn{\Test{$Q$, $\epsilon_0$}}{
\Input{$Q$:queries to the database, $\epsilon_0$: privacy budget.}
NoisyVector $\gets$ $[\ ]$\\
\For{$i=1\dots\text{len(Q)}$}{
	NoisyVector$[i]$ $\gets Q[i]+\text{Lap}(1/\epsilon_0)$
}
\Return NoisyVector
}
\caption{Histogram \label{alg:hist}}
\end{algorithm}

\begin{algorithm}[t]
\SetKwProg{Fn}{function}{\string:}{}
\SetKwFunction{Test}{Histogram}
\SetKwInOut{Input}{input}
\DontPrintSemicolon
\Fn{\Test{$Q$, $\epsilon_0$}}{
\Input{$Q$: queries to the database, $\epsilon_0$: privacy budget.}
NoisyVector $\gets$ $[\ ]$\\
\For{$i=1\dots\text{len(Q)}$}{
    // wrong scale of noise is added
	\fbox{NoisyVector$[i]$ $\gets Q[i]+\text{Lap}(\epsilon_0)$}
}
\Return NoisyVector
}
\caption{Histogram with wrong scale  \label{alg:hist_wrong_scale}}
\end{algorithm}

\subsubsection{Adding Exponential Noise}
One correct variant of Noisy Max  adds $\text{Exponential}(2/\epsilon_0)$ noise, rather than Laplace noise, to each query answer(Algorithm \ref{alg:noisy_max_v2a}). This mechanism has also been proven to be $\epsilon_0$-differential private\cite{dwork2014algorithmic}. Figure \ref{fig:noisy_max_v2a} shows the corresponding test result, which is similar to that of Figure \ref{fig:noisy_max_v1a}. The result indicates that this correct variant likely satisfies $\epsilon_0$-differential privacy for the claimed privacy budget.

\subsubsection{Incorrect Variants of Exponential Noise}

An incorrect variant of $Noisy Max$ has the same setup but instead of returning
the \textit{index} of maximum value, it directly returns the maximum value. 
We evaluate on two variants that report the maximum value instead of the index
(Algorithm \ref{alg:noisy_max_v1b} and \ref{alg:noisy_max_v2b}) and show the
test result in Figure \ref{fig:noisy_max_v1b} and \ref{fig:noisy_max_v2b}. 

For the variant using Laplace noise (Figure \ref{fig:noisy_max_v1b}), we can
see that for $\epsilon_0=0.2$, the line rises at around test $\epsilon$ of 0.4,
indicating that  this algorithm is incorrect for the claimed privacy budget of
$0.2$. The same pattern happens when we set privacy budget to be 0.7 and 1.5:
all lines rise much later than their claimed privacy budget. In this incorrect
version, returning the maximum value (instead of its index) causes the
algorithm to actually satisfy $\epsilon_0 \cdot |Q|/2$ differential privacy
instead of $\epsilon_0$-differential privacy.
For the variant using Exponential noise (Figure \ref{fig:noisy_max_v2b}), the
lines rise much later than the claimed privacy budgets, indicating strong
evidence that this variant is indeed incorrect. Also, we can hardly see the
lines for privacy budgets 0.7 and 1.5, since their p-values remain 0 for all
the test $\epsilon$ ranging from 0 to 2.2 in the experiment.

\begin{figure*}
\begin{subfigure}[t]{0.45\textwidth}
\captionsetup{width=.9\linewidth}
\includegraphics[width=\linewidth]{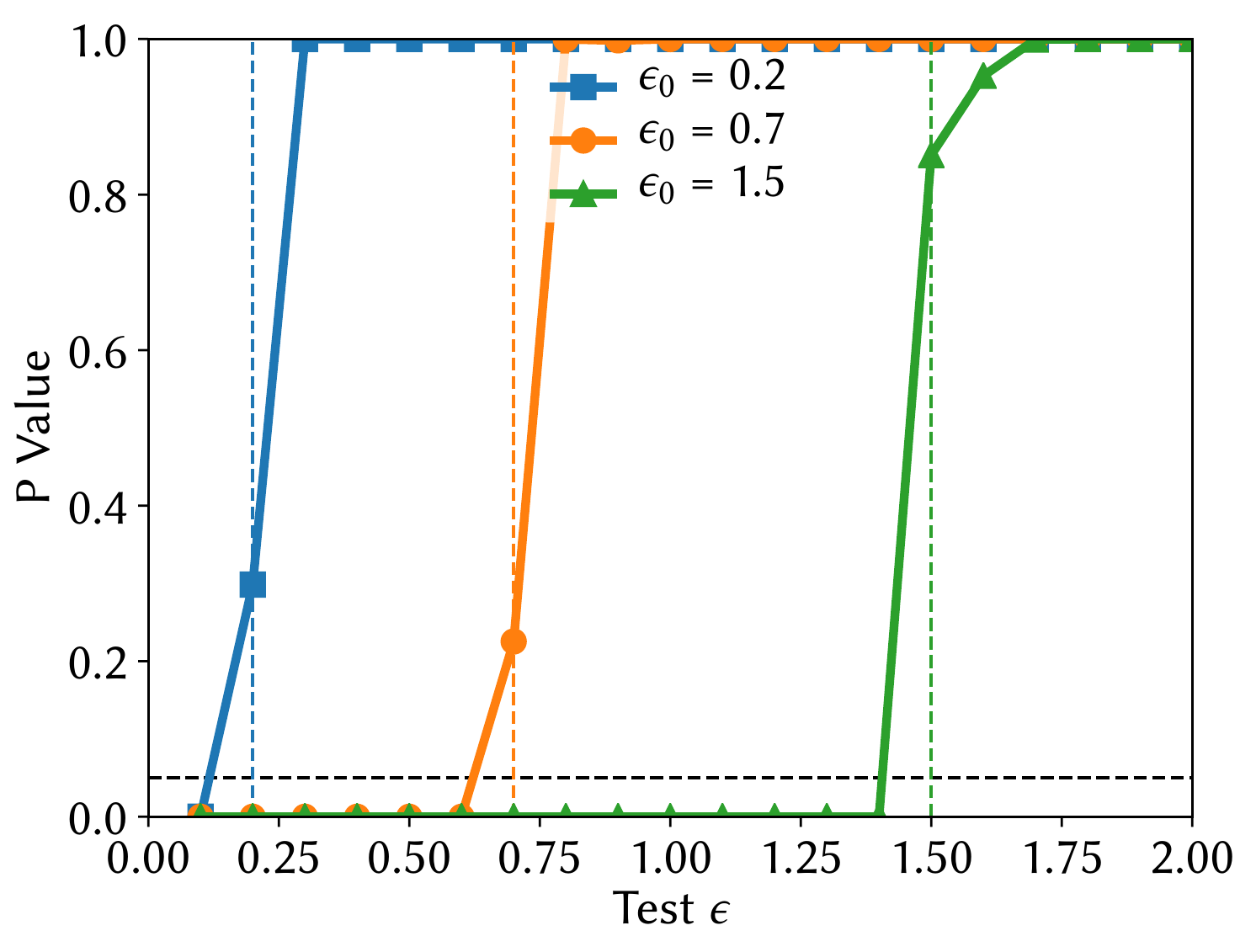}
\caption{Correct Histogram algorithm with Lap($1/\epsilon_0$) noise. 
\label{fig:histogram1/eps}}
\end{subfigure}
\begin{subfigure}[t]{0.45\textwidth}
\captionsetup{width=.9\linewidth}
\includegraphics[width=\linewidth]{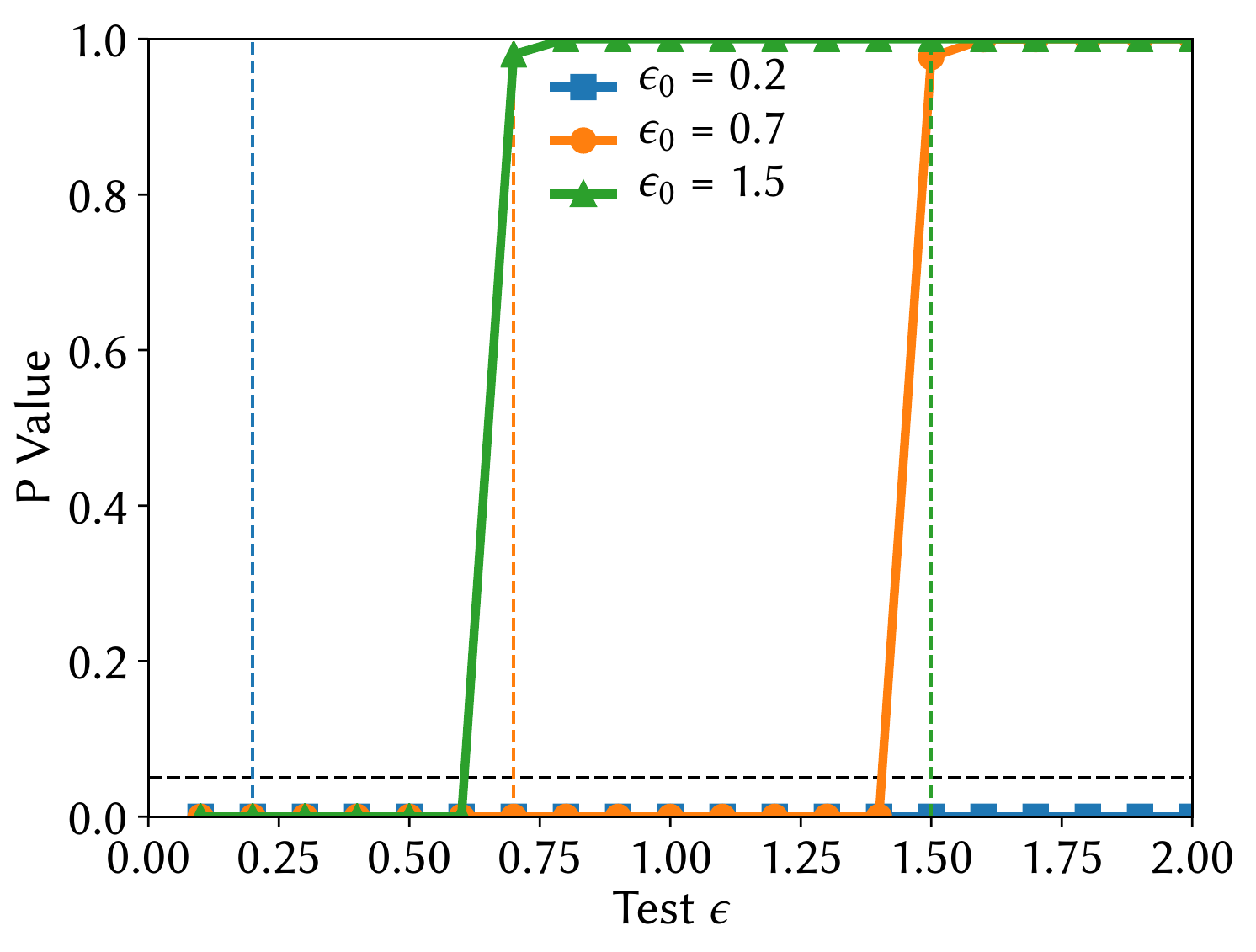}
\caption{Incorrect Histogram algorithm with Lap($\epsilon_0$) noise. It provides more privacy than advertised when $\epsilon_0\geq 1$ and less privacy than advertised when $\epsilon_0<1$.
\label{fig:histogrameps}}
\end{subfigure}
\caption{Results of Histogram algorithm and its variants \label{fig:histogram and laplace}}
\end{figure*}

\subsection{Histogram}
The Histogram algorithm \cite{Dwork:2006:DP:2097282.2097284} is a very simple algorithm for publishing an approximate histogram of the data. The input is a histogram and the output is a noisy histogram with the same dimensions. The Histogram algorithm requires input queries to differ in at most one element. Here we evaluate with different scale parameters for the added Laplace noise. 

The correct Histogram algorithm adds independent $\text{Lap}(1/\epsilon_0)$ noise to each query answer, as shown in Algorithm \ref{alg:hist}. Since at most one query answer may differ by at most 1, returning the maximum value is $\epsilon_0$-differentially private \cite{Dwork:2006:DP:2097282.2097284}.

To mimic common mistakes made by novices of differential privacy, we also evaluate on an incorrect variant where $\text{Lap}(\epsilon_0)$ noise is used in the algorithm (Algorithm \ref{alg:hist_wrong_scale}). We note that the incorrect variant here satisfies $1/\epsilon_0$-differential privacy, rather the claimed $\epsilon_0$-differential privacy.

Figures \ref{fig:histogram1/eps} and \ref{fig:histogrameps} show the test results for the correct and incorrect variants respectively. Here, Figures \ref{fig:histogram1/eps} indicates that the correct implementation satisfies the claimed privacy budgets. For the incorrect variant, the claimed budgets of 0.2 and 0.7 are correctly rejected; this is expected since the true privacy budgets are $1/0.2$ and $1/0.7$ respectively for this incorrect version. Interestingly, the result indicates that for $\epsilon_0=1.5$, this algorithm is likely to be more private than claimed (the line rise around 0.6 rather than 1.5). Again, this is expected since in this case, the variant is indeed $1/1.5 = 0.67$-differentially private.

\subsection{Sparse Vector}
The Sparse Vector Technique (SVT) \cite{Dwork09STOCOnTheComplexity} (see Algorithm \ref{algo:sparse_vector_lyu}) is a powerful mechanism for answering numerical queries. It takes a list of numerical queries and simply reports whether their answers are above or below a preset threshold $T$. It allows the program to output some noisy query answers without any privacy cost. In particular, arbitrarily many ``below threshold'' answers can be returned, but only at most $N$ ``above threshold'' answers can be returned. Because of this remarkable property, there are many variants proposed in both published papers and practical use. However, most of them turn out to be actually not differentially private\cite{lyu2017understanding}. We test our tool on a correct implementation of SVT and the major incorrect variants summarized in \cite{lyu2017understanding}. In the following, we describe what the variants do and list their pseudocodes.

\begin{algorithm}[ht]
\SetKwProg{Fn}{function}{\string:}{}
\SetKwFunction{Test}{SVT}
\SetKwInOut{Input}{input}
\DontPrintSemicolon
\Input{$Q$: queries to the database, $\epsilon_0$: privacy budget\\$T$: threshold, $N$: bound of outputting $True$'s\\$\Delta$: sensitivity}
\Fn{\Test{$Q$, $T$, $\epsilon_0$, $\Delta$, $N$}}{
$out \gets [\,]$\;
$\eta_1 \gets \textrm{Lap}(2*\Delta/\epsilon_0)$\;
$\tilde{T} \gets T + \eta_1$\;
$count \gets 0$\;
\ForEach{$q$ in $Q$}{
$\eta_2 \gets \textrm{Lap}(4*N*\Delta/\epsilon_0)$\;
\eIf{$q+\eta_2\geq \tilde{T}$}
{
$out \gets True::out$\;
$count \gets count + 1$\;
\If{$count \geq N$}{\bfseries Break}
}
{$out \gets False::out$\;}
}
\Return ($out$)
}
\caption{SVT \cite{lyu2017understanding}.\label{algo:sparse_vector_lyu}}
\end{algorithm}

\subsubsection{SVT \cite{lyu2017understanding}}
Lyu et al. have proposed an implementation of SVT and proved that it satisfies $\epsilon_0$-differential privacy. This algorithm (Algorithm \ref{algo:sparse_vector_lyu}) tries to allocate the global privacy budget $\epsilon_0$ into two parts: half of the privacy budget goes to the threshold, and the other half goes to values which are above the threshold. There will not be any privacy cost if the noisy value is below the noisy threshold, in which case the program will output a \textit{False}. If the noisy value is above the noisy threshold, the program will output a \textit{True}. After outputting a certain amount ($N$) of \textit{True}'s, the program will halt. 

Figure \ref{fig:sparse_vector_lyu} shows the test result for this correct implementation. All lines rise around the true privacy budget, indicating that our tool correctly conclude that this algorithm is correct.

\subsubsection{iSVT 1 \cite{stoddard2014differentially}}
One incorrect variant (Algorithm \ref{alg:sparse_vector_stoddard}) adds no noise to the query answers, and has no bound on the number of \textit{True}'s that the algorithm can output. This implementation does not satisfy $\epsilon_0$-differential privacy for any finite $\epsilon_0$. 

This expectation is consistent with the test result shown in Figure \ref{fig:sparse_vector_stoddard}: the p-value never rises at any test $\epsilon$. This result strongly indicates that this implementation with claimed privacy budget 0.2, 0.7, 1.5 is not private for at least any $\epsilon\leq 2.2$.

\begin{algorithm}[t]
\SetKwProg{Fn}{function}{\string:}{}
\SetKwFunction{Test}{iSVT1}
\SetKwInOut{Input}{input}
\DontPrintSemicolon
\Input{$Q$: queries to the database, $\epsilon_0$: privacy budget\\$T$: threshold, $\Delta$: sensitivity}
\Fn{\Test{$Q$, $T$, $\epsilon_0$,  $\Delta$}}{
$out \gets [\,]$\;
$\eta_1 \gets \textrm{Lap}(2 * \Delta/\epsilon_0)$\;
$\tilde{T} \gets T + \eta_1$\;
// no bounds on number of outputs\\
\ForEach{\fbox{$q$ in $Q$}}{
// adds no noise to query answers \\
\fbox{$\eta_2 \gets 0$\; } \\
\eIf{$q+\eta_2\geq \tilde{T}$}
{
$out \gets True::out$\;
}
{$out \gets False::out$\;}
}
\Return ($out$)
}
\caption{iSVT 1 \cite{stoddard2014differentially}. This does not add noise to the query answers, and has no bound on number of \textit{True}'s to output (i.e., $N$). This is  not private for any privacy budget $\epsilon_0$ . \label{alg:sparse_vector_stoddard}}
\end{algorithm}

\begin{algorithm}[t]
\SetKwProg{Fn}{function}{\string:}{}
\SetKwFunction{Test}{iSVT3}
\SetKwInOut{Input}{input}
\DontPrintSemicolon
\Input{$Q$: queries to the database, $\epsilon_0$: privacy budget\\$T$: threshold, $N$: bound of outputting $True$'s\\$\Delta$: sensitivity}
\Fn{\Test{$Q$, $T$, $\epsilon_0$, $\Delta$, $N$}}{
$out \gets [\,]$\;
$\eta_1 \gets \textrm{Lap}(4*\Delta/\epsilon_0)$\;
$\tilde{T} \gets T + \eta_1$\;
$count \gets 0$\;

\ForEach{$q$ in $Q$}{
// noise added doesn't scale with N  \\
\fbox{$\eta_2 \gets \textrm{Lap}(4*\Delta/(3*\epsilon_0))$\; }

\eIf{$q+\eta_2\geq \tilde{T}$}
{
$out \gets True::out$\;
$count \gets count + 1$\;
\If{$count \geq N$}{\bfseries Break}
}
{$out \gets False::out$\;}
}
\Return ($out$)
}
\caption{iSVT 3 \cite{lee2014top}. The noise added to queries doesn't scale with $N$. The actual privacy cost is $\frac{1+6N}{4}\epsilon_0$. \label{alg:sparse_vector_lee}}
\end{algorithm}

\subsubsection{iSVT 2 \cite{chen2015differentially}}
Another incorrect variant (Algorithm \ref{alg:sparse_vector_chen}) has no bounds on the number of \textit{True}'s the algorithm can output. Without the bounds, the algorithm will still output \textit{True} even if it has exhausted its privacy budget. So this variant is not private for any finite $\epsilon_0$.

\begin{algorithm}[t]
\SetKwProg{Fn}{function}{\string:}{}
\SetKwFunction{Test}{iSVT2}
\SetKwInOut{Input}{input}
\DontPrintSemicolon
\Input{$Q$: queries to the database, $\epsilon_0$: privacy budget\\$T$: threshold, $\Delta$: sensitivity}
\Fn{\Test{$Q$, $T$, $\epsilon_0$, $\Delta$}}{
$out \gets [\,]$\;
$\eta_1 \gets \textrm{Lap}(2*\Delta/\epsilon_0)$\;
$\tilde{T} \gets T + \eta_1$\;
// no bounds on number of outputs\;

\ForEach{\fbox{$q$ in $Q$}}{ \;
$\eta_2 \gets \textrm{Lap}(2*\Delta/\epsilon_0)$\;
\eIf{$q+\eta_2\geq \tilde{T}$}
{
$out \gets True::out$\;
}
{$out \gets False::out$\;}
}
\Return ($out$)
}
\caption{iSVT 2 \cite{chen2015differentially}. This one has no bounds on number of \textit{True}'s (i.e, $N$) to output.  This is not private for any finite privacy budget $\epsilon_0$.\label{alg:sparse_vector_chen}}
\end{algorithm}

\begin{algorithm}[t]
\SetKwProg{Fn}{function}{\string:}{}
\SetKwFunction{Test}{iSVT4}
\SetKwInOut{Input}{input}
\DontPrintSemicolon
\Input{$Q$: queries to the database, $\epsilon_0$: privacy budget\\$T$: threshold, $N$: bound of outputting $True$'s\\$\Delta$: sensitivity}
\Fn{\Test{$Q$, $T$, $\epsilon_0$, $\Delta$, $N$}}{
$out \gets [\,]$\;
$\eta_1 \gets \textrm{Lap}(2*\Delta/\epsilon_0)$\;
$\tilde{T} \gets T + \eta_1$\;
$count \gets 0$\;
\ForEach{$q$ in $Q$}{
$\eta_2 \gets \textrm{Lap}(2*N*\Delta/\epsilon_0)$\;
\eIf{$q+\eta_2\geq \tilde{T}$}
{
// output numerical value instead of boolean value\\
\fbox{$out \gets (q+\eta_2)::out$\;}\\
$count \gets count + 1$\;
\If{$count \geq N$}{\bfseries Break}
}
{$out \gets False::out$\;}
}
\Return ($out$)
}
\caption{iSVT 4  \cite{roth2011notes}. When the noisy query answer is above the threshold, output the actual value of noisy query answer.\label{alg:sparse_vector_roth}}
\end{algorithm}

\begin{figure*}[!ht]
\begin{subfigure}[b]{0.45\textwidth}
\captionsetup{width=.9\linewidth}
\includegraphics[width=\linewidth]{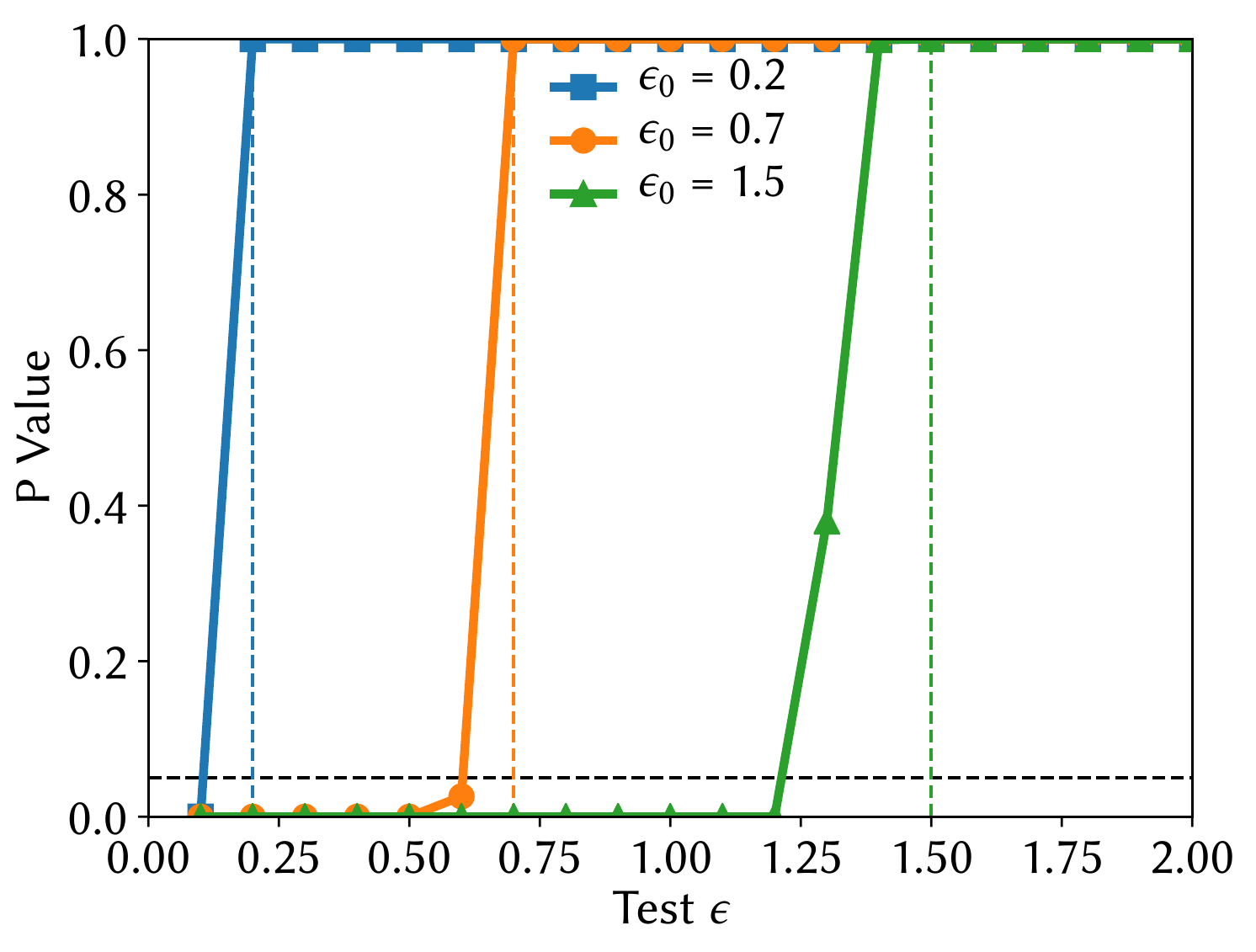}
\caption{Correct implementation of SVT \cite{lyu2017understanding}. \label{fig:sparse_vector_lyu}}
\end{subfigure}
\begin{subfigure}[b]{0.45\textwidth}
\captionsetup{width=.9\linewidth}
\includegraphics[width=\linewidth]{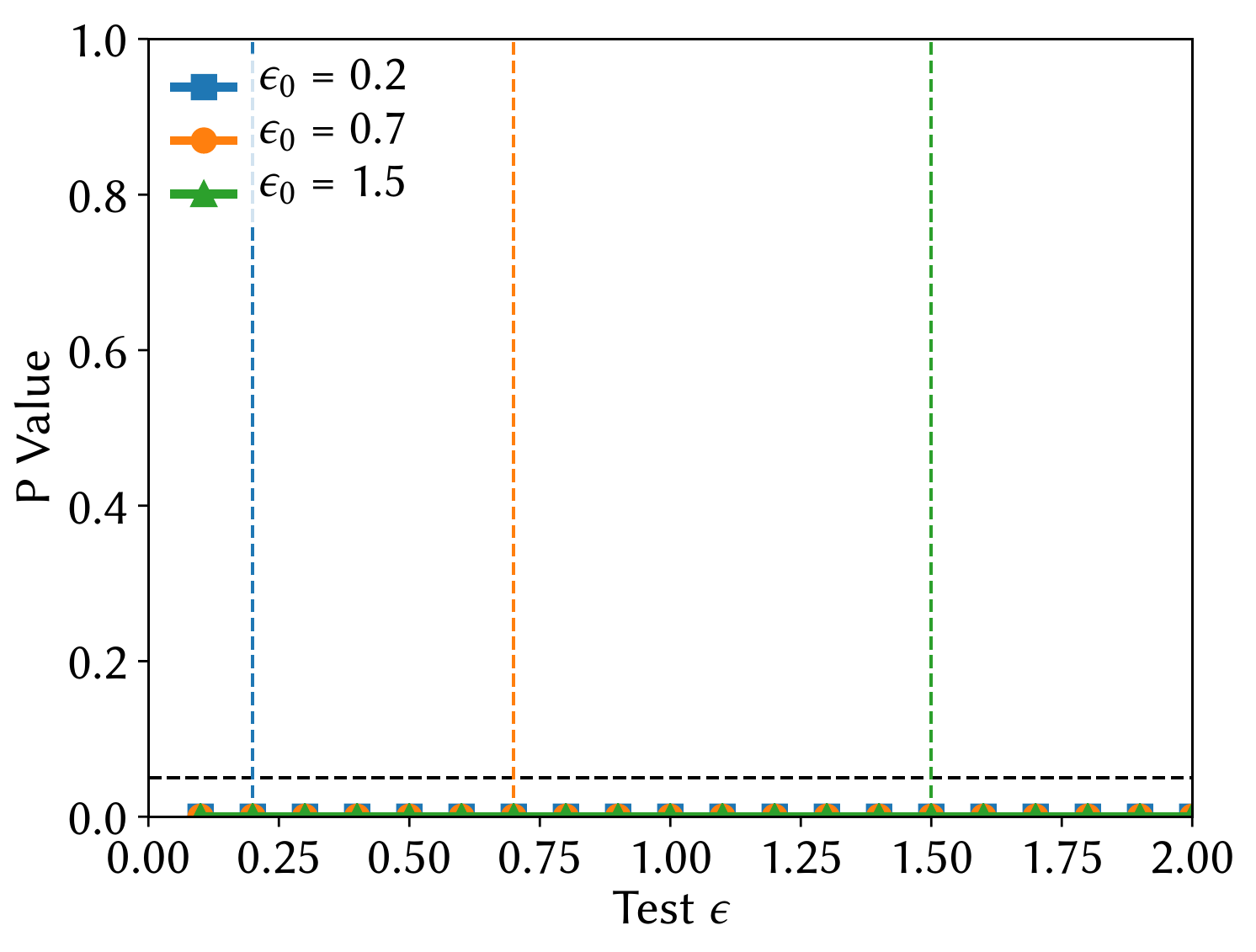}
\caption{iSVT 1 \cite{stoddard2014differentially}\label{fig:sparse_vector_stoddard} adds no noise to query and threshold.}
\end{subfigure}
\begin{subfigure}[b]{0.45\textwidth}
\captionsetup{width=.9\linewidth}
\includegraphics[width=\linewidth]{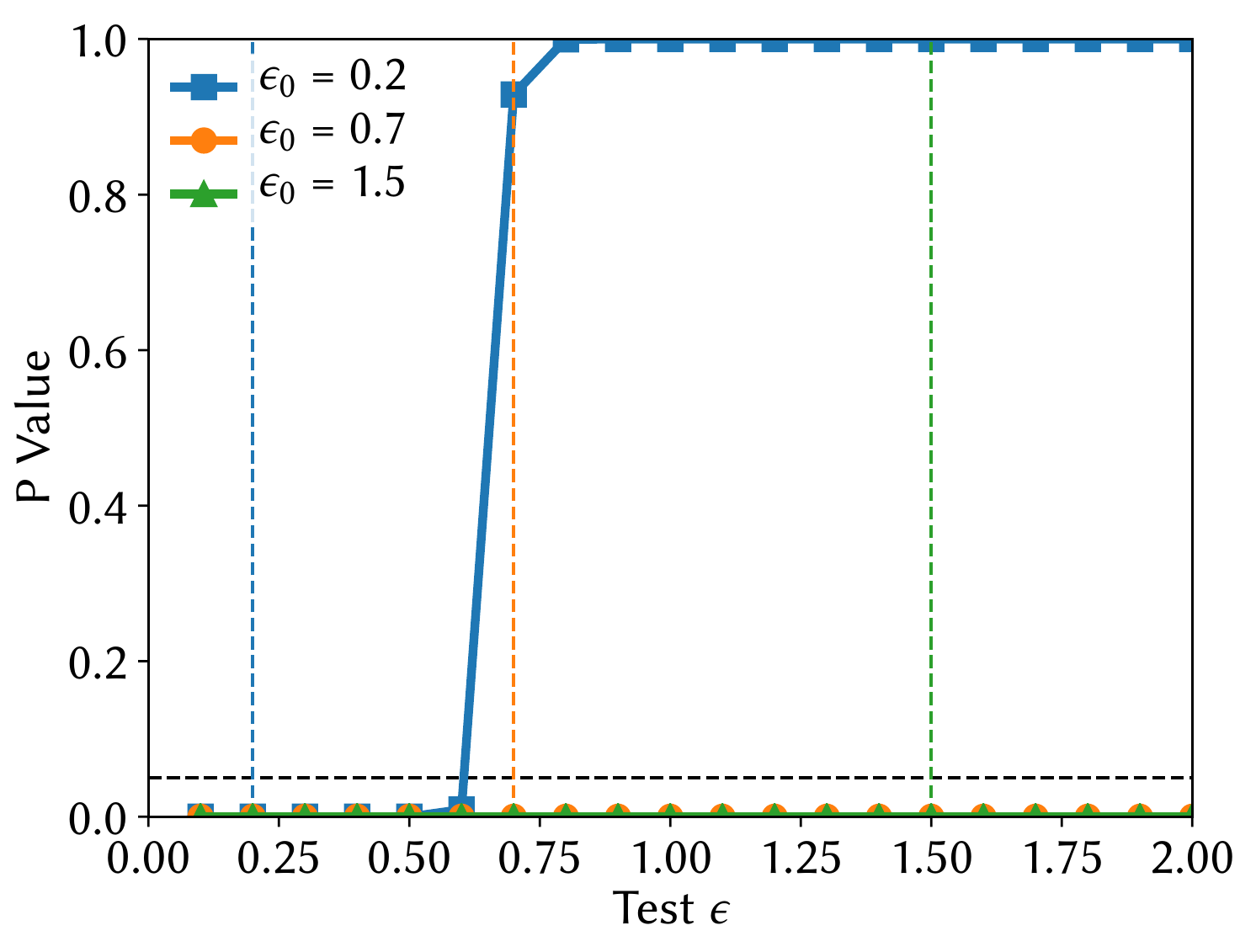}
\caption{iSVT 2 \cite{chen2015differentially} no bounds on outputting \textit{True}'s. \label{fig:sparse_vector_chen}}
\end{subfigure}
\begin{subfigure}[b]{0.45\textwidth}
\captionsetup{width=.9\linewidth}
\includegraphics[width=\linewidth]{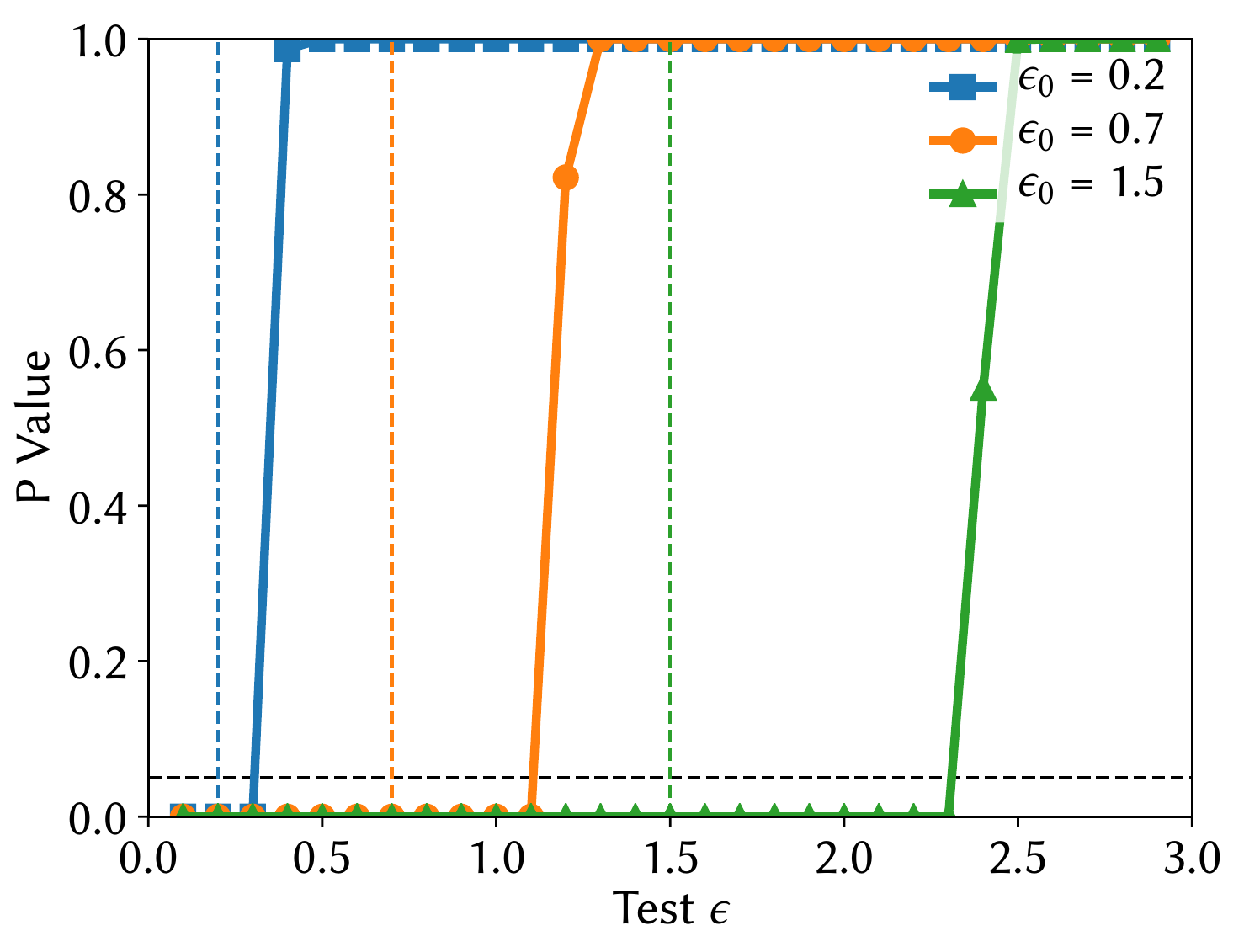}
\caption{iSVT 3 \cite{lee2014top} query noise does not scale with $N$. \label{fig:sparse_vector_lee}}
\end{subfigure}
\begin{subfigure}[b]{0.45\textwidth}
\captionsetup{width=.9\linewidth}
\includegraphics[width=\linewidth]{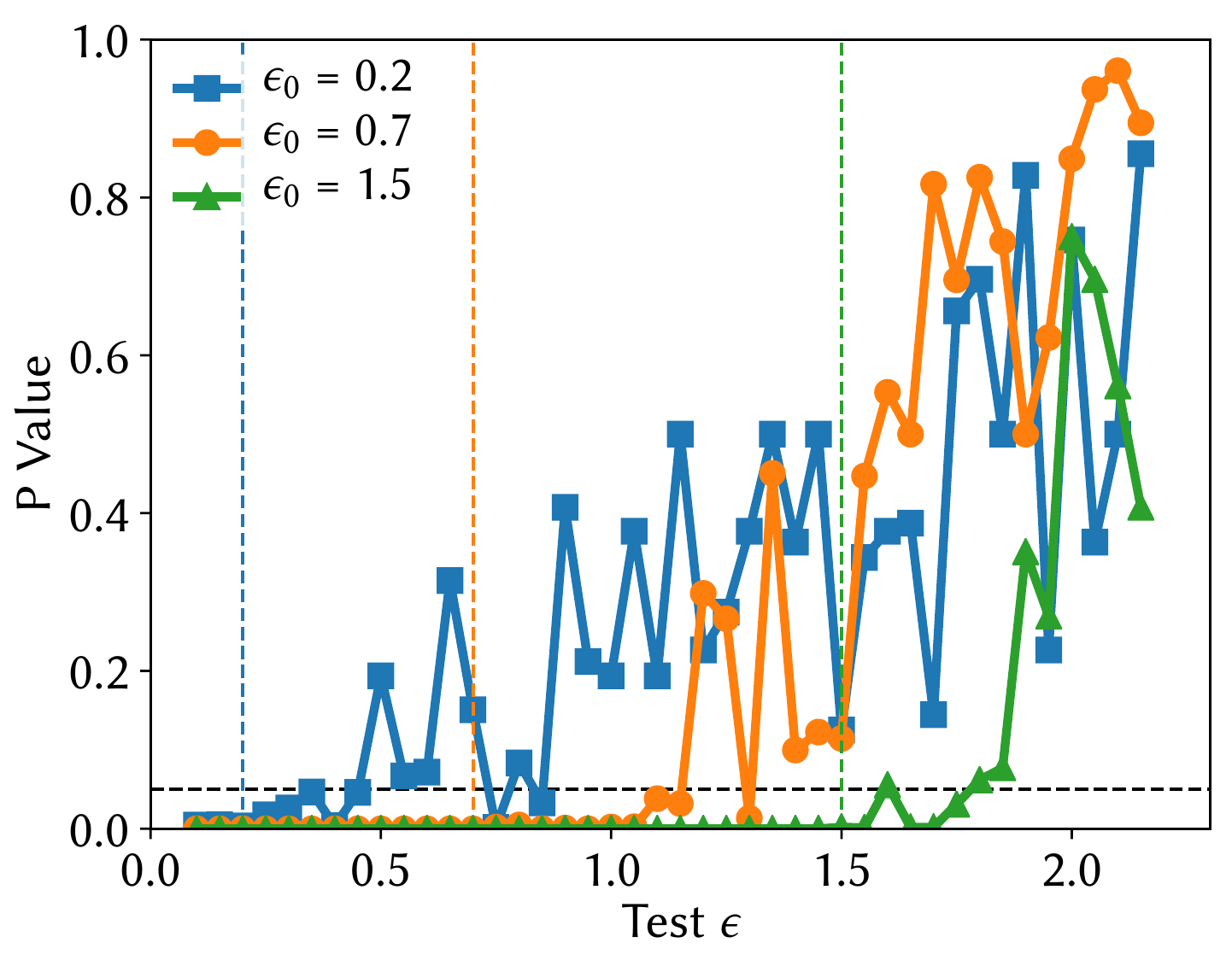}
\caption{iSVT 4 \cite{roth2011notes} outputs the actual query answer when it is above the threshold.\label{fig:sparse_vector_roth}}
\end{subfigure}
\caption{Results for variants of Sparse Vector Technique}
\end{figure*}

\begin{table*}[t]
\caption{Counterexamples detected for incorrect privacy mechanisms\label{tab:selected_queries}}
\small
\begin{center}
\begin{tabular}{c c c c} 
\thickhline
\textbf{Mechanism} ($\epsilon_0=1.5$) & \textbf{Event E}   & \textbf{D1} & \textbf{D2} \\
\hline

Incorrect Noisy Max with Laplace Noise & $\omega \in (-\infty, 0.0)$ & [1, 1, 1, 1, 1] & [0, 0, 0, 0, 0] \\

Incorrect Noisy Max with Exponential Noise & $\omega \in (-\infty, 1.0)$ & [1, 1, 1, 1, 1] & [0, 0, 0, 0, 0]\\
Incorrect Histogram \cite{dwork06Calibrating}& $\omega[0] \in (-\infty, 1.0)$ & [1, 1, 1, 1, 1] & [2, 1, 1, 1, 1] \\ 

iSVT 1 \cite{stoddard2014differentially} & $t(\omega) = 0$ & [1, 1, 1, 1, 1, 1, 1, 1, 1, 1] & [0, 0, 0, 0, 0, 2, 2, 2, 2, 2]\\
iSVT 2 \cite{chen2015differentially} & $t(\omega) = 9$ & [1, 1, 1, 1, 1, 0, 0, 0, 0, 0] & [0, 0, 0, 0, 0, 1, 1, 1, 1, 1] \\
iSVT 3 \cite{lee2014top} & $t(\omega) = 0$ & [1, 1, 1, 1, 1, 0, 0, 0, 0, 0] & [0, 0, 0, 0, 0, 1, 1, 1, 1, 1] \\
iSVT 4 \cite{roth2011notes} & $(\omega.count(False), \omega[9] ) \in \{9\}\times (-2.4, 2.4)$ & [1, 1, 1, 1, 1, 1, 1, 1, 1, 1] & [0, 0, 0, 0, 0, 0, 0, 0, 0, 0]\\
\thickhline
\end{tabular}
\end{center}
\end{table*}

Figure \ref{fig:sparse_vector_chen} indicates this implementation with privacy budget $\epsilon_0 = 0.2$ is most likely \emph{not} private for any $\epsilon \leq 0.5$. When $\epsilon_0 = 0.7 $, we have detected counterexamples showing the algorithm is likely not private for any $\epsilon \in (0, 2.1]$. When $\epsilon_0 = 1.5 $, we have detected counterexamples showing the algorithm is likely not private for any $\epsilon \in (0, 3.0]$.

\subsubsection{iSVT 3 \cite{lee2014top}}
Another incorrect variant (Algorithm \ref{alg:sparse_vector_lee}) adds noise to queries  but the noise doesn't scale with the bound $N$. The actual privacy budget for this variant is $\frac{1+6N}{4}\epsilon_0$ where $\epsilon_0$ is the input privacy budget.

We note that our tool detects the actual privacy cost, as shown in Figure \ref{fig:sparse_vector_lee}, for this incorrect algorithm. Consider privacy budget $\epsilon_0=0.2$. The corresponding line rises at $0.3$, right before the actual budget $\frac{1+6N}{4}\epsilon_0=0.35$ ($N=1$), suggesting the precision of our tool. The same happens for $\epsilon_0=0.7$ and $1.5$. The two lines rise at $1.1$ and $2.3$, which are close to but before the actual budget $1.225$ and $2.625$, respectively.

\subsubsection{iSVT 4 \cite{roth2011notes}}
Another incorrect variant (Algorithm \ref{alg:sparse_vector_roth}) outputs the actual value of noisy query answer when it is above the noisy threshold. 

The interesting part of this algorithm is that, since it outputs heterogeneous list of booleans and values, our event selector chooses $\{9\} \times (-2.4, 2.4)$. This means we choose an event that consists of 9 booleans (in this case, Falses) followed by a number in $(-2.4, 2.4)$. Figure \ref{fig:sparse_vector_roth} shows much noise in it because this one is almost correct in the sense that violations of differential privacy happen with very low probability; thus it is hard to detect its incorrectness. But we can still see that the lines all rise later than the corresponding claimed privacy budget $\epsilon_0$. Hence, our tool correctly concludes that this algorithm does not satisfy $\epsilon_0$-differential privacy. 

\subsection{Performance}
We performed all experiments on a double $\text{Intel}^{\text{\textregistered}}$ $\text{Xeon}^{\text{\textregistered}}$ E5-2620 v4 @ 2.10GHz CPU machine with 64 GB memory. Our tool is implemented in Anaconda distribution of Python 3 and optimized for running in parallel environment to fully utilize the 32 logical cores of the machine.

\begin{table}[t]
\caption{Time spent on running tool for different algorithms \label{tab:performance}}
\begin{tabular}{c c}
\thickhline
\textbf{Mechanism} & \textbf{Time / Seconds} \\
\hline
Correct Laplace Noisy Max\cite{dwork2008differential} & 4.32 \\
Incorrect Laplace Noisy Max & 9.49 \\
Correct Exponential Noisy Max \cite{dwork2008differential} & 4.25 \\
Incorrect Exponential Noisy Max & 8.70 \\
Histogram \cite{Dwork:2006:DP:2097282.2097284}  & 10.39 \\
Incorrect Histogram  & 11.28 \\
SVT \cite{lyu2017understanding} & 1.99 \\
iSVT 1 \cite{stoddard2014differentially} & 1.62 \\
iSVT 2 \cite{chen2015differentially} & 4.56 \\
iSVT 3 \cite{lee2014top}& 2.56 \\
iSVT 4 \cite{roth2011notes} & 22.97\\
\thickhline
\end{tabular}
\end{table}

For each test $\epsilon$, we set the samples of iteration $n$ to be 500,000 for the hypothesis test and 100,000 for the event selector and query generator. Table \ref{tab:performance} lists the average time spent on hypothesis test for a specific test $\epsilon$ (i.e., the average time spent on generating one single point in the figures) for each algorithm. The results suggest that it is very efficient to run a test for an algorithm against one privacy cost: all tests finish within 23 seconds.

The time difference between Noisy Max, Histogram and Sparse Vector Technique is due to the nature of the algorithms. For SVT, the parameter $N$ is set to 1, meaning that the algorithm will halt once it hit a True branch. For Noisy Max and Histogram, all noise will be calculated and applied to each query answer, consuming more time to calculate p-values. Another factor that will also influence the test time is the search space of events. Correct \textit{Noisy Max} returns an index which we would have a search space of only integers ranging from 1 to the length of queries. However, the incorrect \textit{Noisy Max} will return a real number so the search space would be much larger than the correct one, thus taking more time to find a suitable event $E$. This also occurs in Sparse Vector Technique.

\section{Conclusions and Future Work}\label{sec:conc}
While it is invaluable to formally verify correct differentially-private algorithms, we believe that it is equally important to detect incorrect algorithms and provide counterexamples for them, due to the subtleties involved in algorithm development. We proposed a novel semi-black-box method of evaluating differentially private algorithms, and providing counterexamples for those incorrect ones. We show that within a few seconds, our tool correctly rejects incorrect algorithms (including published ones) and provides counterexamples for them. 

Future work includes extensions that detect violations of differential privacy even if those violations occur with extremely small probabilities. This will require additional extensions such as a more refined use of program analysis techniques (including symbolic execution) that reason about what happens when a program is run on adjacent databases. Additional extensions include counterexample generation for other variants of differential privacy, such as approximate differential privacy, zCDP, and renyi-differential privacy.

\begin{acks}
We thank anonymous CCS reviewers for their helpful suggestions.  This work was
partially funded by NSF awards \#1228669, \#1702760 and \#1566411.
\end{acks}

\bibliographystyle{ACM-Reference-Format}
\bibliography{references,diffpriv}
\clearpage
\appendix

\end{document}